%% file: main.tex
\documentclass[a4paper]{article}
\usepackage{a4wide}
\usepackage[UKenglish]{babel}
\usepackage{graphicx,subfigure,caption}
\usepackage{algorithm,algorithmic}
\graphicspath{{./}{images/}}
\usepackage[colorlinks=true,linkcolor=blue,citecolor=red]{hyperref}
\usepackage{xcolor}
\usepackage{amssymb,amsthm,empheq,bbold}
\usepackage[inline]{enumitem}
\usepackage{calrsfs}\DeclareMathAlphabet{\pazocal}{OMS}{zplm}{m}{n}
\usepackage{authblk}
\usepackage{comment}
\usepackage[normalem]{ulem}
\usepackage{booktabs}
\usepackage{csquotes}
\usepackage[
    backend=biber,
    style=numeric-comp,
    doi=false,
    url=false,
    isbn=false,
    maxcitenames=3,
    maxbibnames=2,
    giveninits=true
]{biblatex}
\AtEveryBibitem{\clearlist{language}}
\renewbibmacro{in:}{}
\usepackage{tikz}
\usetikzlibrary{shadows}

\DeclareMathOperator{\sgn}{sgn}

\newcommand{\R}{\mathbb{R}}

\newtheorem{assumption}{Assumption}[section]

\newtheorem{theorem}[assumption]{Theorem}
\newtheorem{lemma}[assumption]{Lemma}
\newtheorem{corollary}[assumption]{Corollary}
\theoremstyle{remark}
\allowdisplaybreaks

\title{From graphons to real-world networks: kinetic opinion dynamics under selective media influence}
\author{Bertram Düring$^{\ast,\star}$}
\author{Martina Fraia$^\dagger$}
\author{Alessandro Licciardi$^{\dagger,\ddagger}$}
\affil{{\small $\ast$ Mathematics Institute, University of Warwick, Coventry, United Kingdom\\$\dagger$ Department of Mathematical Sciences ``G. L. Lagrange'', Politecnico di Torino, Turin, Italy\\
$\ddagger$ Istituto Nazionale di Fisica Nucleare (INFN), Sezione di Torino, Turin, Italy\\
$\star$ Author for correspondance (bertram.during@warwick.ac.uk)
}}
\date{}
		
\begin{document}
\maketitle
	
\begin{abstract}
\noindent We propose a kinetic model of opinion dynamics under selective media influence on both graphon-based and real-world networks. The media action, inspired by Hallin's theory of spheres, is incorporated through a model predictive control strategy designed to steer agents' opinions toward a desired target opinion. For the resulting Boltzmann-type description, we analyse the evolution of the moments and, in the quasi-invariant interaction limit, derive a Fokker--Planck-type equation together with a characterisation of its stationary states. We also prove exponential convergence to equilibrium in the Fourier metric. Numerical experiments are performed on networks generated by a Gaussian graphon and on real-world, single-issue Twitter networks, allowing us to investigate the role of control and interaction parameters, as well as the impact of the subset of agents subject to media influence. Using real-world social networks data to initialise opinions and infer the interaction structure, we then compare the dynamics obtained on the original networks with those produced by Gaussian graphons fitted to their adjacency matrices, thereby assessing the descriptive power of the graphon approach for real-world opinion dynamics.

\medskip


\medskip

\end{abstract}

\section{Introduction}

Mathematical models of opinion dynamics have received growing attention in the past two decades.
Most models consider a large number of interacting agents whose microscopic, agent-level dynamics lead to the formation of complex macroscopic states, and, depending on the model, show effects like clustering, consensus or polarisation.

Kinetic models for opinion formation, which build on mathematical tools from statistical mechanics lead to generalisations of the classical Boltzmann equation for gas dynamics (mesoscopic scale) and approximate Fokker-Planck-type equations in quasi-invariant limits (macroscopic scale), see e.g.\  \cite{toscani2006kinetic,During:strongleaders,motsch14,zbMATH06350891,pareschi2013BOOK,During:inhomogeneous,TTZ18,albi2016optimal,zbMATH06553232,zbMATH06684855,zbMATH06684856,zbMATH06712347,garnier17,zbMATH06875734,zbMATH06892489,zbMATH06977672,zbMATH07205672,zbMATH06944946,albi2019boltzmann,zbMATH07146312,aylaj2020unified,bellomo2021whatislife,boscheri2021modeling,burger2021network,During:polling,during2026voter}. Alternative approaches in which agents are modelled by coupled stochastic differential equations lead to very similar macroscopic partial differential equations in a mean-field limit, see, e.g. \cite{garnier17,GGSP,wang2017noisy,Nugent_Gomes_Wolfram_2025}.

In this paper we are considering a kinetic model for opinion formation for a large number of interacting agents whose opinions are affected in two ways: (i) agent-to-agent binary opinion dynamics of bounded confidence type \cite{hegselmann02,toscani2006kinetic}, and (ii) media influence via a model predicative control which aims to shape agent's opinions towards a desired target opinion.

Mathematically, the closest kinetic model to our  opinion dynamics is \cite{APZ14}. There, a  kinetic opinion formation model with leaders is proposed, where the leaders’ strategy is determined by an optimisation problem based on an objective functional aimed at steering followers toward consensus. Similar to our paper, the minimisation of a cost functional is embedded into the microscopic interactions through a model predictive control.
We note that a PDE-constrained optimisation variation on the controlled leader-follower theme has been investigated in \cite{during2025pde} which formulates the leader–follower kinetic model as a PDE-constrained optimal control problem with a Fokker–Planck-type system as the state equation. A number of recent opinion formation papers using the mean-field approach \cite{coculescu2024opinion,helfmann2023modelling} also
explicitly model media and influencer dynamics.

There a two key differences between our approach and the above.
We consider explicitly the influence of the structure of the social network on the opinion dynamics using a {\it graphon} that describes the network structure in the continuum limit. Moreover, the action of the media in our model is  different from previous approaches, since we assume that the media can interact only with a subset of agents whose opinions lie in a set $S \subseteq [-1,1]$. The targeted action of the media on a subset of agents is motivated by the sociological media and political communication theory of {\it Hallin's spheres}. Independent of Hallin's theory it is reasonable to assume that media can only affect part of the discourse, as their influence will be more limited for extreme opinions because these are often tied to certain identities or insulated communities.

Hallin \cite{hallin1986} divides the world of political discourse into three concentric spheres which are schematically depicted in Figure~\ref{fig:spheres}, consensus, legitimate controversy, and deviance:\\
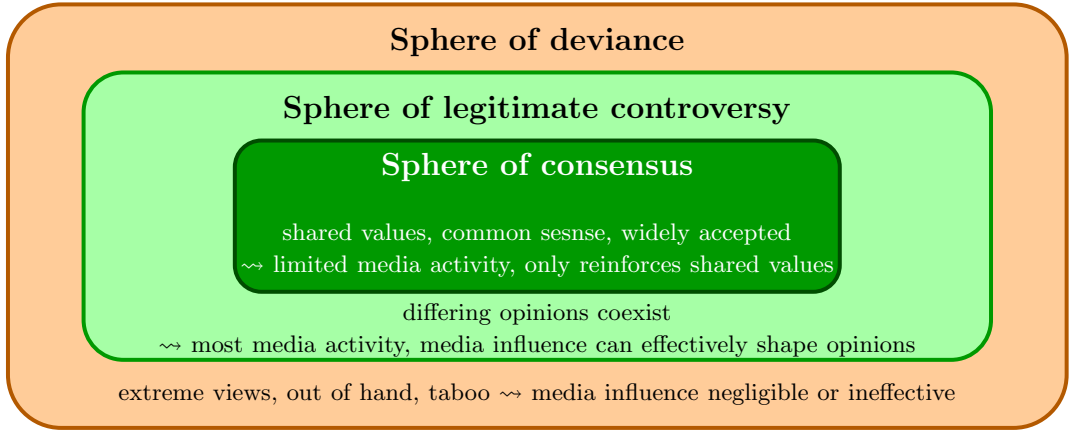
\begin{figure}
\centering
\begin{tikzpicture}

  \draw[
    rounded corners=20pt,
    fill=orange!40,
    draw=orange!70!black,
    line width=1.5pt
  ]
  (-7, -2.8) rectangle (7, 2.8);
  \node at (0, 2.3) {\large\textbf{Sphere of deviance}};
  \node[align=center, text width=14cm] at (0, -2.35)
    {\small extreme views, out of hand, taboo $\leadsto$ media influence negligible or ineffective};
  \draw[
    rounded corners=15pt,
    fill=green!35,
    draw=green!60!black,
    line width=1.5pt
  ]
  (-6.0, -1.9) rectangle (6.0, 1.9);
  \node at (0, 1.4) {\large\textbf{Sphere of legitimate controversy}};
  \node[align=center, text width=12cm] at (0, -1.5)
    {\small differing opinions coexist\\ $\leadsto$ most media activity, media influence can effectively shape opinions};

  \draw[
    rounded corners=10pt,
    fill=green!60!black,
    draw=green!30!black,
    line width=1.5pt
  ]
  (-4, -1.0) rectangle (4, 1.0);
  \node[white] at (0, 0.65) {\large\textbf{Sphere of consensus}};
  \node[white, align=center, text width=8cm] at (0, -0.45)
    {\small shared values, common sesnse, widely accepted\\ $\leadsto$ limited media activity, only reinforces shared values};
\end{tikzpicture}
\caption{Schematic view of Hallin's spheres of consensus, legitimate controversy and deviance.}
\label{fig:spheres}
\end{figure}
{\it Sphere of consensus.}
This region contains opinions that are widely accepted. In this case media communication does not aim at persuasion or opinion change, but rather reinforces already shared values and the action of the media does not produce a significant dynamical effect on the distribution of opinions hence we can consider just the classical opinion dynamics of consensus.\\
{\it Sphere of deviance.}
This sphere includes opinions that are considered out of hand for public discourse or taboo. Such views are typically ignored or explicitly discredited by the media and any control action becomes negligible or ineffective.\\
{\it Sphere of legitimate controversy.}
Between these two extremes lies the sphere of legitimate controversy, where differing opinions coexist within an admissible range. This is the region in which media activity is most relevant. In this regime, media influence can effectively modify opinions, as agents are neither fully aligned nor completely detached from the mainstream discourse.

As a consequence, our idea is to allow the targeted effect of the media only on the subset of the sphere of legitimate controversy and outside that subset solely consider the effect of interaction between agents.

To represent the network in our kinetic model, we adopt the {\it graphon} approach as recently discussed in \cite{franceschi2022,taricco2026homogeneousboltzmanntypeequationsdense}. In earlier works on kinetic models of opinion formation, various authors incorporated network features into the kinetic description by introducing the degree of an agent (or node) as an additional independent variable, see for example
\cite{albi2016optimal} which proposes a kinetic description of opinion dynamics on time-dependent large-scale networks, combining binary opinion interactions with creation and removal of links in the network, or
\cite{TTZ18} which studies kinetic models of opinion formation on social networks where the state depends on both opinion and agent connectivity.

In recent years, the continuum limit of large discrete networks has been explored via the so-called {\it graphon}, see \cite{borgs2018,lovasz12,cambridge-book,caron23,glasscock2019graphon}. In the continuum limit the adjacency matrix of the discrete network is replaced by a function $W(x,y)$ depending on the continuous graph positions (or labels) $x,y\in[0,1]$. The graphon has only recently been introduced in kinetic Boltzmann-type models of opinion formation in \cite{nurisso23,franceschi2022}. In \cite{nurisso23} a $0$-$1$-valued interaction kernel appears in the Boltzmann-type equation which models the presence or absence of edges (connections) between nodes (agents) in the network. In \cite{franceschi2022} the network is incorporated explicitly through graphons which feature as kernel in the Boltzmann-type equation. The authors in \cite{franceschi2022} also derive the corresponding limiting Fokker–Planck equation and study kinetic opinion formation on power-law, small-world, and $k-NN$ graphon structures.
See also \cite{taricco2026homogeneousboltzmanntypeequationsdense} for a recent, more detailed exposition of using the graphon in kinetic opinion formation models. For related research on mean-field equations on graphons we refer to \cite{aletti22,bonnet22,bayraktar23,coppini22}. 

This development is a milestone, as graphons allow to elevate kinetic models for opinion formation from {\it all-to-all} interactions (as is standard in kinetic theory) to preferential interactions described by connections between agents in the social network in an elegant and mathematically sound way. Hence, the graphon approach allows to advance these models towards a more realistic description of the real opinion dynamics which in reality crucially depends on the structure of the social network.

With the soundness and effectiveness of the graphon approach in kinetic models for opinion formation established \cite{franceschi2022,taricco2026homogeneousboltzmanntypeequationsdense}, an important question is whether the graphon has sufficient descriptive power to capture the richness of opinion dynamics as in real-world social networks. Consequently, one of the aims of this paper is to turn to this important question.

The main contributions of this paper are:
\begin{itemize}
    \item we propose a new kinetic model for opinion dynamics under selective media influence on a network represented by a graphon, and derive the governing Boltzmann-type and Fokker-Planck-type equations,
    \item we analyse the model in terms of moment evolution and equilibration,
    \item we discuss and characterise steady states of the Fokker-Planck equation,
    \item we study numerically the opinion dynamics using the graphon network description as well as the same dynamics evolving on real-world social networks, and
    \item we investigate numerically the descriptive power of the graphon approach by fitting synthetic Gaussian graphons to real networks' adjacency matrices.
\end{itemize}
This paper is structured as follows: In Section~\ref{sec:micro} we introduce the model for opinion dynamics under selective media influence. The action of the media is inspired by the theory of Hallin's spheres and uses a model predictive control approach, discussed in detail in Section~\ref{sec:MPC} to steer agents' opinions towards a desired target opinion. We analyse the moments of the Boltzmann-type equation in Section~\ref{sec:momentevol}. In the quasi-invariant limit, we derive a Fokker-Planck-type equation in Section~\ref{sec:FP} and discuss the characterisation of its stationary solution in Section~\ref{sec:stationary}. In Section~\ref{sec:equilibration} we prove a result on the exponential convergence of solutions in the Fourier metric. Results of our numerical experiments are reported in Section~\ref{sec:numerics}. First, we show opinion dynamics under selective media influence in a range of scenarios with the network given by a Gaussian graphon, and explore the effects of control and interaction parameters as well as the role of the set in which the media effect is present. Second, we turn to real-world social network data in the form of single-issue Twitter networks. We use the real data to initialise agents' opinions and determine the network structure, before running the opinion dynamics of our model with selective media influence. Finally, we fit a synthetic Gaussian graphons to the real networks' adjacency matrices, and compare the resulting dynamics in both settings, to explore the descriptive power of the graphon approach.

\section{Microscopic opinion formation model}
\label{sec:micro}
At the microscopic (or agent) level, we consider a finite size system of $N$ agents. The agents form a social network with the agents corresponding to the nodes and their connections to the edges of a graph. This social network (graph) is given and fixed in time, the connections (edges) are weighted with weights in $[0,1]$. The state of any two agents is characterised by two variables: their continuous opinions $w, w_*\in [-1,1]$ and their (fixed) labels $i,j \in \{1,\dots,N\}$ in the social network.

The evolution of agents' opinions is due to  two effects:
\begin{itemize}
\item
binary interaction between connected agents (corresponding to connected nodes in the graph) which is driven by compromise combined with random fluctuations, as in opinion formation models of bounded confidence type, and
\item
interaction between agents and the media.
\end{itemize}
We discuss both effects separately in the following.

The opinion dynamics between agents is of bounded confidence type. Following \cite{toscani2006kinetic,pareschi2013BOOK} we consider that agents change their opinions instantaneously through binary interaction, i.e. when two agents with pre-interaction opinions $w, w_*$ interact, their post-interaction opinions $w', w_*'$ are given by 
\begin{align}
\label{post-int-op1}
w'&= w + \gamma P(w,w_*)(w_*-w) + \eta D(w),\\
\label{post-int-op2}
w_*'&= w_* + \gamma P(w_*,w)(w-w_*) + \eta_* D(w_*),
\end{align}
where $\gamma \in (0,1/2)$ is a given compromise parameter and $\eta, \eta_*$ are random variables with zero mean and variance $\sigma^2 > 0$, modelling the random fluctuations. This random change is modulated by the function $D(\cdot)$ which ensures that post-interaction opinions remain within the interval $[-1,1]$, see Appendix~\ref{app:bounded}. A common choice \cite{toscani2006kinetic,During:strongleaders,During:inhomogeneous} for this function is 
$$
D(w)=(1-w^2)^\alpha,
$$
for some $\alpha >0.$
The compromise function $P(\cdot, \cdot)$ is localising compromise which is a key feature of bounded confidence models. Typical choices for $P$ \cite{toscani2006kinetic,During:strongleaders,During:inhomogeneous} are characteristic functions of the set $\{|w-w_*|\le r\}$ with some given radius $r$ or a hyperbolic tangent function. Alternatively, choosing $P\equiv 1$ allows compromise between agents across the whole interval $[-1,1].$

The change of agents' opinions due to interaction with the media is restricted to agents with opinions within a given subset $S\subset [-1,1]$ as suggested by Hallin's spheres of controversy discussed above, i.e.\ $S$ represents the sphere of legitimate controversy. The action of the media is modelled as a time-dependent control $z$ which aims to drive opinions towards a given desired target opinion $w_d$. 
An agent's post-interaction opinion after interacting with the media is given by
\begin{align}
\label{post-int-media}
w''= w + \delta \chi_{\{w \in S\}} z, 
\end{align}
where $\delta >0 $ represents the strength of the influence of the media and  
$$\chi_{\{w \in S\}} = \begin{cases}
    1, \text{ if } w\in S,\\
    0, \text{ otherwise,}
\end{cases}$$ is the characteristic function that restricts the action of the media to the sphere of controversy or another subset of opinions. 
This localised, external media influence is reminiscent of the interaction with a background (thermal bath) in standard kinetic theory. It leads to new and non trivial interplay between the endogenous opinion interactions among agents and exogenous media control. In particular, depending on the strength of the control of the media and the size and shape of $S$, this effect can counteract against the symmetry of classical compromise models, similar as seen in \cite{franceschi2022} where an external effect is introduced to break the consensus in the system.

The control $z(t)$ in \eqref{post-int-media} is chosen via a model predictive control approach, the details of which are discussed in the next section.

\section{Model predictive control action of the media}
\label{sec:MPC}
In this section we present a Model Predictive Control (MPC) approach to model the action of the media on the system of agents. Having in mind the notion of Hallin's spheres as discussed in the introduction, it is reasonable to consider $S$ as the sphere of legitimate controversy and define it as a symmetric interval 
\begin{align}
\label{eq:choiceS}
S = [-r_2,-r_1] \cup [r_1,r_2], \qquad 0 \le r_1 < r_2 \le 1.
\end{align}
The sphere of consensus is then associated with $(-r_1,r_1)$, while $[-1,-r_2)\cup (r_2,1]$ is the sphere of deviance, see Figure~\ref{fig:intervals}.
Alternatively, other subsets of $[-1,1]$ or the whole interval $[-1,1]$ can be chosen as the set $S$.
\begin{figure}
    \centering\begin{tikzpicture}

  \pgfmathsetmacro{\rone}{0.2}   
  \pgfmathsetmacro{\rtwo}{0.7}   
  \pgfmathsetmacro{\sc}{5}       

  \pgfmathsetmacro{\xA}{\sc}           
  \pgfmathsetmacro{\xB}{\rtwo*\sc}     
  \pgfmathsetmacro{\xC}{\rone*\sc}     
  \pgfmathsetmacro{\bh}{0.5}           


  \filldraw[fill=orange!40, draw=orange!70!black, line width=0.8pt]
    ({-\xA},{-\bh}) rectangle ({-\xB},{\bh});

  \filldraw[fill=green!35, draw=green!60!black, line width=0.8pt]
    ({-\xB},{-\bh}) rectangle ({-\xC},{\bh});

  \filldraw[fill=green!60!black, draw=green!30!black, line width=0.8pt]
    ({-\xC},{-\bh}) rectangle ({\xC},{\bh});

  \filldraw[fill=green!35, draw=green!60!black, line width=0.8pt]
    ({\xC},{-\bh}) rectangle ({\xB},{\bh});

  \filldraw[fill=orange!40, draw=orange!70!black, line width=0.8pt]
    ({\xB},{-\bh}) rectangle ({\xA},{\bh});

  \foreach \x in {-\xA, -\xB, -\xC, 0, \xC, \xB, \xA} {
    \draw[black, line width=0.6pt] ({\x},{-\bh}) -- ({\x},{-\bh-0.22});
  }

  \node[below, font=\small] at ({-\xA},{-\bh-0.22}) {$-1$};
  \node[below, font=\small] at ({-\xB},{-\bh-0.22}) {$-r_2$};
  \node[below, font=\small] at ({-\xC},{-\bh-0.22}) {$-r_1$};
  \node[below, font=\small] at (0,       {-\bh-0.22}) {$0$};
  \node[below, font=\small] at ({\xC},  {-\bh-0.22-0.04}) {$r_1$};
  \node[below, font=\small] at ({\xB},  {-\bh-0.22-0.04}) {$r_2$};
  \node[below, font=\small] at ({\xA},  {-\bh-0.22}) {$1$};

  \pgfmathsetmacro{\midOL}{(-\xA-\xB)/2}
  \pgfmathsetmacro{\midIL}{(-\xB-\xC)/2}
  \pgfmathsetmacro{\midOR}{(\xA+\xB)/2}
  \pgfmathsetmacro{\midIR}{(\xB+\xC)/2}

  \node[above, font=\small\bfseries] at ({\midOL},{\bh +0.07}) {Deviance};
  \node[above, font=\small\bfseries] at ({\midIL},{\bh}) {Controversy};
  \node[above, font=\small\bfseries] at (0,        {\bh +0.07}) {Consensus};
  \node[above, font=\small\bfseries] at ({\midIR}, {\bh}) {Controversy};
  \node[above, font=\small\bfseries] at ({\midOR}, {\bh +0.07}) {Deviance};
  \node[font=\small\bfseries] at ({\midIL}, 0) {$S$};
  \node[font=\small\bfseries] at ({\midIR}, 0) {$S$};
  \pgfmathsetmacro{\arrowstart}{\xA}
  \pgfmathsetmacro{\arrowend}{\xA+1}
  \draw[->, >=stealth, line width=0.5pt]
    ({\arrowstart}, -\bh) -- ({\arrowend}, -\bh)
    node[right, font=\small] at (\xA+0.5,-\bh-0.22){opinion $w$};
\end{tikzpicture}
\caption{Opinion interval $[-1,1],$ divided into spheres of consensus, deviance and legitimate controversy.}
\label{fig:intervals}
\end{figure}
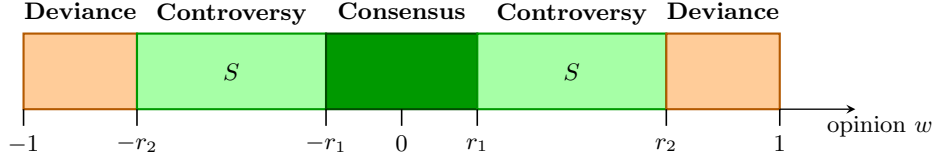

The time-dependent control $z$ is chosen by solving a constrained optimisation problem. The media has the aim to steer the agents' opinions towards a desired opinion $w_d \in [-1,1]$. The choice of the cost functional reflects the trade-off between the alignment of the agents' opinions with the target opinion $w_d$ and the cost associated with the media intervention (for example related to the advisory costs, payments of influencers in the social network, etc.). Hence in the cost functional we consider the sum of two terms: a term that penalises the deviation of the post-interaction opinion $w''$ from the desired opinion $w_d$ in a least-squares sense and a second term, that represents a penalisation of the control effort, in which the parameter $\tilde{\lambda} > 0$ balances the cost of the intervention. In particular, larger values of $\tilde{\lambda} > 0$ correspond to more costly media actions leading to a weaker control, while smaller values favour more powerful interventions. Mathematically, the penalty term prevents unrealistic large controls and guarantees well-posedness of the minimisation problem.

Moreover, differently from model predictive control approaches employed previously for kinetic models of opinion formation (see, e.g. \cite{franceschi2022}), we introduce a constraint which ensures that the controlled opinions remain within the region of action of the media $S$. It is natural to assume that the media, which only acts in $S$, want the post-interaction opinions to remain in $S$ so that them remain within the media's influence.
Mathematically, the choice of the best media action becomes a convex optimisation problem with non linear inequality constraints.
Specifically, in each interaction \eqref{post-int-media}, the control is chosen by solving the following constrained optimisation problem:
\begin{align}
\label{eq:opt_problem_lambda}
z^* =
\begin{cases}
\displaystyle
\arg\min_{z\in \mathbb{R}} J(z)
=
\frac{1}{2}(w'' - w_d)^2
+
\frac{1}{2}\tilde{\lambda} z^2,
\\[6pt]
w'' = w + \delta z,
\\[6pt]
r_1 \le |w''| \le r_2,
\end{cases}
\end{align}
where $w_d \in S \subset [-1,1]$ is the desired opinion and $\tilde{\lambda} > 0$ is the penalty coefficient for the cost of the media intervention.

We notice that the quadratic structure of both terms in $J$ defined in \eqref{eq:opt_problem_lambda} ensures that the cost functional is strictly convex in $z$, which implies the uniqueness of the solution in the unconstrained setting. This property is crucial in view of deriving explicit solutions for the media control.

Exploiting the media post-interaction opinion \eqref{post-int-media} in the first equation in \eqref{eq:opt_problem_lambda} the cost functional becomes:
\begin{align*}
J(z) = \frac{1}{2}(w+\delta z-w_d)^2 +
\frac{1}{2}\tilde{\lambda} z^2,
\end{align*}
subject to
\begin{align}
\label{eq:constrain_square}
r_1^2 \le (w+\delta z)^2 \le r_2^2.
\end{align}
which is equivalent to the inequality constraint in  \eqref{eq:opt_problem_lambda}. 
We notice that the constraints introduced within the optimisation problem ensure that $w''$ remains within a subset of $[-1,1]$, and therefore it is not necessary to introduce any upper bound for $\delta$  to obtain a value of $w''$ that lies within the admissible set of opinions.

To solve this constrained convex problem, we consider the \textit{Karush--Kuhn--Tucker (KKT) conditions}, which provide necessary and sufficient optimality conditions due to the convexity of both the objective functional and the admissible set.
This method introduces Lagrange multipliers associated with the constraints and leads to a system of equations and complementary conditions that identify whether the optimal solution lies in the interior of the admissible set or on its boundary. This structure is particularly advantageous, as it leads to an explicit characterisation of the optimal control in terms of a projection of the unconstrained solution onto the admissible region. Therefore, the KKT approach not only guaranties optimality, but also provides a constructive way to compute the control in closed form.

\begin{lemma}
\label{lem:wellposed_kkt}
Let $w, w_d \in S \subseteq [-1,1]$, $\delta > 0$, and $\tilde{\lambda} > 0$.
Under the hypothesis that the admissible set,
\begin{align*}
\mathcal{A} \coloneqq \{ z \in \mathbb{R} : r_1^2 \le (w+\delta z)^2 \le r_2^2 \},
\end{align*}
is not empty, consider the constrained optimisation problem
\begin{align}
\begin{cases}
\arg\min_{z \in \mathbb{R}} \;
J(z) = \frac{1}{2}(w+\delta z - w_d)^2 +
\frac{1}{2}\tilde{\lambda} z^2,\\
r_1^2 \le (w+\delta z)^2 \le r_2^2,
\label{eq:opt_problem}
\end{cases}
\end{align}
with $0 \le r_1 < r_2 \le 1$. Then Problem~\eqref{eq:opt_problem} admits a unique minimiser $z^* \in \mathbb{R}$ which is given by
\begin{align*}
z^* = \frac{1}{\delta}\Big(\mathrm{Proj}[r_1,r_2](|w''_{\mathrm{free}}|) \sgn(w''_{\mathrm{free}}) -w\Big)
\end{align*}
with
$$ w''_{\mathrm{free}} =\frac{\tilde{\lambda}w + \delta^2w_d}{\tilde{\lambda} + \delta^2}.$$
\end{lemma}

\begin{proof}

The cost functional is quadratic and satisfies
\begin{align*}
J''(z) = \delta^2 + \tilde{\lambda} > 0,
\end{align*}
hence $J$ is strictly convex on $\mathbb{R}$. Therefore, any minimiser, if it exists, is unique.

Since the mapping $z \mapsto w+\delta z$ is affine and continuous, and $\mathcal{A}$ is a non-empty, closed subset of $\mathbb{R}$ then $\mathcal{A}$ is compact. Since the functional $J$ is continuous, then the existence of a minimiser follows from the Weierstrass theorem:
$$\exists z_* \in \mathcal{A} \text{ such that } J(z_*)= \min_{\mathcal{A}} J.$$

Rewriting the constraints as
\begin{align*}
g_1(z) &= (w+\delta z)^2 - r_2^2 \le 0, \\
g_2(z) &= r_1^2 - (w+\delta z)^2 \le 0,
\end{align*}
we can introduce the Lagrange multipliers $\mu_1,\mu_2$ associated with the upper and lower inequality constraint, respectively, and rewrite Problem~\eqref{eq:opt_problem} as follows:
\begin{align*}
\begin{cases}
\delta(w+\delta z^* - w_d)+\tilde{\lambda} z^*+2\delta(\mu_1 - \mu_2)(w+\delta z^*)= 0,\\
g_1(z^*) \le 0, \quad g_2(z^*) \le 0,\\
\mu_1 \ge 0, \quad \mu_2 \ge 0,\\
\mu_1 g_1(z^*) = 0, \quad \mu_2 g_2(z^*) = 0.
\end{cases}
\end{align*}
If the unconstrained minimiser
\begin{align*}
z_{\mathrm{free}} = \frac{\delta(w_d - w)}{\tilde{\lambda} + \delta^2}
\end{align*}
satisfies the constraint, then both multipliers $\mu_1, \mu_2$ vanish and $z^* = z_{\mathrm{free}}$ with
$$ w''_{\mathrm{free}} = w+ \delta z_{\mathrm{free}} = w +  \frac{\delta^2(w_d - w)}{\tilde{\lambda} + \delta^2} = \frac{\tilde{\lambda}w + \delta^2w_d}{\tilde{\lambda} + \delta^2}.$$
Otherwise, at least one constraint is active, and the optimal solution lies on the boundary of the admissible set, i.e.
\begin{align*}
|w+\delta z^*| \in \{r_1, r_2\}.
\end{align*}
Thus, we can conclude that the optimal post-interaction opinion $w''_*$ is the projection of $w_{\mathrm{free}}$ on the subset $S$.

In the case of the symmetric definition \eqref{eq:choiceS} of the subset $S$ as according to Hallin's sphere theory we have 
\begin{align}
\label{eq:opt_w''}
w''_* = \frac{\mathrm{Proj}[r_1,r_2](|w''_{\mathrm{free}}|)}{|w''_{\mathrm{free}}|}w''_{\mathrm{free}},
\end{align}
in which the $|\cdot|$ represents the absolute value and the projection is defined as:
\begin{align*}
\mathrm{Proj}[r_1,r_2](x)=
\begin{cases}
r_1, \quad  \text{ if } x<r_1,\\
x, \quad  \text{ if } r_1 \leq x \leq r_2,\\
r_2, \quad  \text{ if } x>r_2.
\end{cases}
\end{align*}
We note that all calculations remain valid if we consider a general subset $S$ that is not necessarily the above chosen symmetric version of $S$ given in \eqref{eq:choiceS}, and then we can just write the optimal post-interaction opinion as:
$$w''_* = \mathrm{Proj}[S](|w''_{\mathrm{free}}|)\sgn(w''_{\mathrm{free}}).$$
To conclude, going back to the symmetric definition of $S$ given in \eqref{eq:choiceS}, the optimal control is defined as 
\begin{align*}
z^*=\frac{1}{\delta}
\left(w''_*-w \right),
\end{align*}
which is equivalent to
\begin{align*}
z^* &= \frac{1}{\delta}\left(
\mathrm{sgn}(w''_{\mathrm{free}})\, \max\left\{r_1,\,\min\left\{|w''_{\mathrm{free}}|,\,r_2\right\}\right\}-w\right)\\
&= \frac{1}{\delta}\Big(\frac{\mathrm{Proj}[r_1,r_2](|w''_{\mathrm{free}}|)}{|w''_{\mathrm{free}}|}w''_{\mathrm{free}} -w\Big)\\
&= \frac{1}{\delta}\Big(\mathrm{Proj}[r_1,r_2](|w''_{\mathrm{free}}|) \sgn(w''_{\mathrm{free}}) -w\Big).
\end{align*}
This provides the complete characterisation of the optimal control.
\end{proof}

\section{Mesoscopic description and evolution of the first moment}
\label{sec:momentevol}

For a very large number of interacting agents, using standard methods of kinetic theory, we can now derive the Boltzmann-type equation associated with the microscopic interactions \eqref{post-int-op1}, \eqref{post-int-op2}, \eqref{post-int-media}.
The state of the system is then described by the probability density function
$$f = f(w,x,t),$$
where $w \in [-1,1]$ denotes the opinion variable and $x \in \Omega \subseteq [0,1]$ represents the continuous label (or position) on the graph. The quantity $f(w,x,t)\,dw\,dx$ gives the fraction of agents at time $t$ with opinion in $[w,w+dw]$ and position in $[x,x+dx]$.
The time evolution of $f$, for any smooth test function $\varphi = \varphi(w,x)$, is governed by a Boltzmann-type equation in weak form:
\begin{align}
\frac{d}{dt}\int_0^1 \int_{-1}^1 \varphi(w,x) f(w,x,t) \, dw \, dx 
= \frac{1}{\tau_0}\bigl ( \mathcal{Q}_0(f,f), \varphi \bigr ) + \frac{1}{\tau_1} \bigl ( \mathcal{Q}_1(f), \varphi \bigr ),
\label{eq:boltz}
\end{align}
The operator $\mathcal{Q}_0$ is associated with the evolution due to the binary opinion interaction between agents. The kernel of interaction for this opinion dynamics is expressed by a graphon $W(x,y)$.

The \textit{graphon} is a Lebesgue-measurable function $$W:\Omega^2 \subseteq[0,1]^2 \to \R_+,$$ that constitutes a powerful instrument to describe the limits of large dense networks \cite{borgs2018,franceschi2022,glasscock2019graphon}. From a practical point of view, $W(x,y)$ can be interpreted as the continuous analogue of the adjacency matrix and represents the probability of interaction between two agents having labels (or positions) $x$ and $y$, respectively. The use of graphons in kinetic theory provides a more general description than earlier approaches to networks, in which each agent was associated with a random variable corresponding to their connectivity (e.g.\ the number of connections) \cite{albi2024datadriven,albi2016optimal,albi17,LoyRaviolaTosin_2022}. For more details, we refer the reader to \cite{borgs2018,franceschi2022,glasscock2019graphon}.

The operators (in weak form) on the right hand side of \eqref{eq:boltz} are given by:
\begin{align}
\bigl ( \mathcal Q_0(f,f), \varphi \bigr )
&=
\int_0^1 \int_0^1 \int_{-1}^1 \int_{-1}^1 
W(x,y)\big\langle\varphi(w') - \varphi(w)\big\rangle 
 f(w,x,t) f(w_*,y,t) \, dw \, dw_* \, dx \, dy,
 \label{eq:q_0}
\end{align}
where $\langle\cdot\rangle$ denotes taking the expectation with respect to the random variables occurring in \eqref{post-int-op1},\eqref{post-int-op2} and 
\begin{align}
\bigl ( \mathcal Q_1(f), \varphi \bigr )
=
\int_0^1 \int_{-1}^1 
\big(\varphi(w''_*) - \varphi(w)\big) f(w,x,t) \, dw \, dx,
\label{eq:q_1}
\end{align}
in which $w'$ and $w''_*$ are the post-interaction opinions \eqref{post-int-op1} and \eqref{post-int-media}, with the optimal media policy as given in Section~\ref{sec:MPC}. 

From equation \eqref{eq:boltz} we can analyse the different moments of the system. 
We define the local density
\begin{align*}
\rho(x,t) := \int_{-1}^1 f(w,x,t)\,dw.
\end{align*}
Choosing $\varphi(w,x) = 1$, we immediately obtain
\begin{align*}
\frac{d}{dt}\int_0^1 \rho(x,t)\,dx = 0,
\end{align*}
since both interaction terms vanish due to $\varphi(w')= \varphi(w'')=\varphi(w)$. Therefore, the total mass is conserved in time.

As is typical in bounded confidence opinion formation models, the first moment, the mean opinion, is not conserved in time in general, see e.g.\ \cite{toscani2006kinetic,During:strongleaders,pareschi2013BOOK}. However, under mild symmetry condition on the graphon $W$, we have the following result.
Note that symmetry of $W$ holds if the social network  consists of a non-directed graph.

\begin{lemma}
\label{lem:mean_control}
Let $f = f(w,x,t) \in L^1([-1,1]\times \Omega \times \R_+)$ be a probability density and $W$ symmetric, i.e.\ $W(x,y)=W(y,x)$ for all $x,y\in\Omega\subseteq[0,1]$. Denote by
\begin{gather*}
m(t)  \coloneq \int_0^1 \int_0^1 w f(w,x,t)\,dw\,dx,\quad
M_S(t)  \coloneq\int_0^1 \int_S f(w,x,t)\,dw\,dx, \\
m_S(t)  \coloneq \int_0^1 \int_S w f(w,x,t)\,dw\,dx,
\end{gather*}
the total opinion mean of the system and the mass and the opinion mean in the subset $S$, respectively.
If the compromise function of the opinion exchange $P$ in \eqref{post-int-op1},\eqref{post-int-op2} is symmetric, i.e.\ $P(w,w_*) = P(w_*,w)$ for all $w,w_* \in [-1,1]$, and 
\begin{align*}
w_d = \frac{m_S(t)}{M_S(t)},
\end{align*} then the total mean opinion of the system $m(t)$ is conserved in time.
\end{lemma}

\begin{proof}

To evaluate the evolution of the mean we can replace $\varphi(w,x) = w$ in \eqref{eq:boltz} and find: 
\begin{align*}
\partial_t m(t)
&= \frac{1}{\tau_0}
\int_0^1 \int_0^1 \int_{-1}^1 \int_{-1}^1 
W(x,y) {\langle w' - w\rangle} f(w,x,t) f(w_*,y,t)\,dw\,dw_*\,dx\,dy \nonumber\\
&\quad + \frac{1}{\tau_1} \int_0^1
\int_{-1}^1 (w''_* - w) f(w,x,t)\,dw\,dx.
\end{align*}
Using the symmetry assumptions
$P(w,w_*) = P(w_*,w)$
and $W(x,y)= W(y,x)$,
a standard change of variables shows that 
\begin{align*}
\int_0^1 \int_0^1\int_{-1}^1 \int_{-1}^1 W(x,y)\langle w' - w \rangle f(w,x) f(w_*,y)\,dw\,dw_*\,dx\,dy  = 0,
\end{align*}
so that binary interactions do not contribute to the evolution of the mean opinion.

We now focus on the control term. Restricting to the admissible set $S$, we obtain
\begin{align*}
\int_0^1 \int_S (w''_* - w) f(w,x,t)\,dw\,dx.
\end{align*}

Using the explicit form of the optimal control in the unconstrained regime, we have
\begin{align*}
w'' - w = \delta z = \frac{\delta^2}{\tilde{\lambda} + \delta^2}(w_d - w),
\end{align*}
which yields
\begin{align*}
\int_0^1 \int_S (w'' - w) f(w,x,t)\,dw\,dx =
\frac{\delta^2}{\tilde{\lambda} + \delta^2}
\int_0^1 \int_S (w_d - w) f(w,x,t)\,dw\,dx.
\end{align*}

Using the definitions given, we immediately recognise the mean and the mass in the 
\begin{align}
\int_0^1 \int_S (w'' - w) f(w,x,t)\,dw\,dx
=
\frac{\delta^2}{\tilde{\lambda} + \delta^2}
\big(w_d M_S(t) - m_S(t)\big).
\label{eq:evolution_media}
\end{align}
We notice that we have a conservation of the total mean if and only if: 
$$w_d = \frac{m_S(t)}{M_S(t)}. $$
 This condition corresponds to a neutral control regime, in which the media does not introduce any net drift in the global mean opinion.
\end{proof}
The previous lemma highlights the selective nature of the control mechanism. Only the fraction of agents with opinions in the set $S$ contributes to the evolution of the mean opinion.
In particular, equation \eqref{eq:evolution_media} suggests that the rate of convergence towards the desired opinion $w_d$ depends on the mass $M_S(t)$ of agents lying in the sphere of controversy. If this fraction is small, the overall effect of the media is weak and the mean opinion evolves slowly. Conversely, when a large portion of the population belongs to $S$, the alignment toward $w_d$ is significantly raised.

Moreover, the contribution is modulated by the distance from the target opinion: agents whose opinions are farther from $w_d$ exert a stronger influence on the shift of the mean. This results, in a scenario in which the binary opinion dynamics is omitted and only the action of the media is present, in a collective dynamics in which the global mean opinion is driven toward $w_d$, with an effective strength determined by the distribution of agents within $S$.

In contrast, when $w_d \neq m_S(t)/M_S(t)$, the media induces a systematic shift of the mean towards the target opinion $w_d$. The strength of this effect is proportional to the mass $M_S(t)$ of agents in $S$.

\section{Fokker-Planck-type equation in the quasi-invariant limit}
\label{sec:FP}
The analytical asymptotic-in-time distribution of the Boltzmann-type equation is very complex to deduce. For this reason, several alternative models have been proposed to investigate the large-time behaviour. As presented in \cite{toscani2006kinetic} and \cite{pareschi2013BOOK} a possible approach is to consider a particular regime: the \textit{quasi-invariant limit}. The idea is to work in a regime where both the interaction and the stochastic contributions are rescaled to be very small, and consider this for large times. This is consistent with the fact that we are interested in the behaviour of our system near equilibrium, where each interaction produces only a small variation of the microscopic state. Under this assumption, one can formally derive a Fokker-Planck-type equation that preserves the main  properties of the microscopic dynamics close to the stationary state and allows for an easier characterisation of the asymptotic distribution $f^{\infty}$.

Starting from \eqref{eq:boltz}--\eqref{eq:q_1} which is satisfied by the probability density $f(w,x,t)$, we consider the following scaling:
$$ t \rightarrow  \frac{t}{\gamma}, \quad \delta \rightarrow k \gamma, \quad \Tilde{\lambda}\rightarrow \nu \gamma$$
and perform a Taylor expansion, that is consistent with the fact that in the chosen regime $w'-w$ is small:
\begin{align*}
\phi(w')- \phi(w) &= \phi'(w)(w'-w)+ \frac{1}{2}\phi''(w)(w'-w)^2 + \Tilde {R}(\gamma, \sigma^2)
\end{align*}
and similarly for $w''.$

For streamlining the presentation we consider $P(w_*,w)\equiv 1$ for the derivation of the limit equations.
Dividing by $\gamma$ and passing to the limit $\gamma, \sigma \to 0$ with $\sigma^2/\gamma = \lambda$ fixed and defining:
\begin{align}
A_W[f](x,w) & \coloneqq \int_0^1\int_{-1}^1 W(x,y)(w_* - w)f(w_*,y)\,dw_*\,dy, \label{eq:A_w}\\ 
H(w) &\coloneqq (w_d-w)\chi_{\{w\in S\}}.\label{eq:H}
\end{align}
we obtain the following Fokker-Planck-type equation (written in strong form): 
\begin{align}
\partial_t f =
\partial_w
\Bigg[\frac{1}{\tau_0}A_W[f](x,w) f -
\frac{\lambda}{2}\partial_w(D^2(w)f)+\frac{k^2}{\nu\tau_1}H(w)f
\Bigg],
\label{eq:fokker_planck}
\end{align}
which is subject to no-flux boundary conditions, see \cite{toscani2006kinetic, pareschi2013BOOK,During:strongleaders} for details of the derivation.

To obtain the stationary behaviour we look at the equilibrium imposing $\partial_tf^{\infty}=0$ coupled to the null flux boundary conditions:
\begin{align}
\begin{cases}
\partial_w
\Bigg[\frac{1}{\tau_0}A_W[f^{\infty}](x,w) f^{\infty} -
\frac{\lambda}{2}\partial_w(D^2(w)f^{\infty})+\frac{k^2}{\nu\tau_1}H(w)f^{\infty}
\Bigg] =0 \text{ in } I \times \Omega,\\ 
\Big(\frac{1}{\tau_0}A_W[f^\infty]+\frac{k^2}{\nu\tau_1}H(w)\Big)f^\infty-\frac{\lambda}{2}\partial_w(D^2(w)f^\infty)=0 \text{ on } \partial I \times \partial\Omega.
\end{cases}
\end{align}
We define:
\begin{align}
m_W(x) & \coloneqq \int_0^1 W(x,y)m^{\infty}(y)\,dy,\qquad
\rho_W(x) \coloneqq \int_0^1 W(x,y)\rho(y)\,dy,
\label{eq:mass_media_weighted}
\end{align}
in which $m^{\infty}(y)$ denotes the asymptotic mean of the opinion of the agents with position equal to $y$ and $\rho(y)$ is the mass of the agents with position equal to $y$.
It is now trivial to observe that evaluating the operator $A_W$ defined in \eqref{eq:A_w} in the steady state:
\begin{align*}
A_W[f^\infty](x,w) = m_W(x) - w \rho_W(x).
\end{align*}
Setting $g(w,x)= f^{\infty}(w,x) D^2(w)$, we obtain the following ODE for the stationary profile:
\begin{align*}
\frac{\partial_w g}{g}
= \frac{2}{\lambda}\frac{\frac{1}{\tau_0}(m_W(x)-w\rho_W(x))+\frac{k^2}{\nu\tau_1}H(w)}{D^2(w)}.
\end{align*}
Integrating both sides, we end up with:
\begin{align*}
g(w,x) = C(x)\exp\Bigg(\frac{2}{\lambda}\int_0^w\frac{\frac{1}{\tau_0}(m_W(x)-v\rho_W(x))+\frac{k^2}{\nu\tau_1}H(v)}{D^2(v)}\,dv\Bigg).
\end{align*}
Hence, choosing $D(w)= (1- w^2)^\alpha$, our stationary solution must satisfy: 

\begin{align}
f^\infty(w,x)= \frac{C(x)}{(1-w^2)^{2 \alpha}}\exp\Bigg(\frac{2}{\lambda}\int_0^w\frac{\frac{1}{\tau_0}(m_W(x)-v\rho_W(x))+\frac{k^2}{\nu\tau_1}H(v)}{(1- v^2)^{2 \alpha}}\,dv\Bigg).
\label{eq:f_infinity}
\end{align}
The diffusion coefficient is chosen as
\[
D(w) = (1-w^2)^\alpha, \qquad \alpha > 0,
\]
which vanishes at the boundary opinions $w = \pm 1$, reflecting the realistic
assumption that extreme opinions are harder to change. The exponent $\alpha$
controls the rate of degeneracy near the boundaries and plays a crucial role
in determining whether the admissible interval $[-1,1]$ can be preserved across all interactions.
As shown in Appendix \ref{app:bounded}, the case $\alpha \ge 1$ ensures the existence of a
fixed support for the noise variable $\eta$, independent of the current
opinion $w$, under which both post-interaction opinions $w'$ and $w_*'$
remain in $[-1,1]$. The canonical choice $\alpha = 1$, originally proposed
by Toscani in \cite{toscani2006kinetic}, yields the sufficient uniform condition
$|\eta| \le (1-\gamma)/2$ and leads to a Fokker--Planck equation with an
explicitly computable stationary state.
For $0 < \alpha < 1$, a fixed support is no longer sufficient, since the
required bound on $|\eta|$ degenerates as $|w| \to 1$. Nevertheless, this
regime remains accessible provided one allows the support of $\eta$ to depend
on the current opinion $w$ chosen so that
\begin{equation*}
|\eta|
\;\le\;
(1-\gamma)\,
\frac{(1-|w|)^{1-\alpha}}{(1+|w|)^{\alpha}}.
\end{equation*}
This opinion-dependent truncation prevents boundary violations at the cost
of a more complex interaction kernel. For a detailed discussion, see Appendix~\ref{app:bounded}.

\section{Characterisation of the stationary distribution}
\label{sec:stationary}
In this section we discuss the characterisation of the stationary solution of the Fokker-Planck-type equation \eqref{eq:fokker_planck},
\begin{align*}
f^\infty(w,x)= \frac{C(x)}{(1-w^2)^{2\alpha}}
\exp\left(\frac{2}{\lambda} \int_0^w
\frac{\dfrac{1}{\tau_0}(m_W(x)-v\rho_W(x))
+\dfrac{k^2}{\nu\tau_1}H(v)}
{(1-v^2)^{2\alpha}}\,dv\right).
\end{align*}
First of all, we remark that a distribution of Beta type on the interval $(-1,1)$ is characterised by a
density of the form
\begin{align}
f^{\infty}(w) \propto (1-w)^{a(x)-1}(1+w)^{b(x)-1},
\label{eq:beta}
\end{align}
that is, a purely algebraic behaviour at the boundaries $w=\pm 1$.
Therefore, in order for the stationary state to belong to this class, the
exponential term must produce only logarithmic contributions, so that it can
be rewritten as powers of $(1\pm w)$.

As already shown in \cite{toscani2006kinetic,pareschi2013BOOK} when $\alpha=1/2$ the denominator simplifies as
$(1-v^2)^{2\alpha}=(1-v^2)$ and splitting the integrand and computing explicitly:
\begin{equation*}
\int \frac{dv}{1-v^2} = \frac{1}{2}\log\!\left(\frac{1+v}{1-v}\right),
\qquad
\int \frac{v\,dv}{1-v^2} = -\frac{1}{2}\log(1-v^2).
\end{equation*}
Hence, in the pure opinion dynamics (i.e. $H(v)=0$), we obtain
\begin{align}
f^\infty(w,x)
= \frac{C(x)}{1-w^2}
\exp\!\Bigg(
\frac{1}{\lambda\tau_0}
\Big[
m_W(x)\log\!\left(\frac{1+w}{1-w}\right)
+\rho_W(x)\log(1-w^2)
\Big]
\Bigg).
\end{align}
Exploiting the properties of
the logarithm, we have
\begin{align}
f^\infty(w,x)
= C(x)\,
(1-w)^{a(x)-1}
(1+w)^{b(x)-1},
\label{eq:f_beta}
\end{align}
where
\begin{equation}
a(x) = \frac{\rho_W(x)-m_W(x)}{\lambda\tau_0},
\qquad
b(x) = \frac{\rho_W(x)+m_W(x)}{\lambda\tau_0}.
\label{eq:ab_beta_op}
\end{equation}
Thus, in the absence of the media term, the stationary distribution is a
Beta-type distribution, since the interaction is modulated by the graphon, the parameters are entirely determined by the local graphon-weighted mass $\rho_W(x)$ and mean opinion $m_W(x)$.

Let us now turn to the media term, $H(v) = (w_d -v)\chi_{\{v \in S\}}$, with $S=[-r_2,-r_1]\cup[r_1,r_2]$. Because $H$ is discontinuous at $\partial S = \{-r_2,-r_1,r_1,r_2\}$, we no longer solve the ODE on the whole of $(-1,1)$ at once: we solve it separately on each of the five subintervals on which $H$ is smooth,
$$
(-1,-r_2),\quad [-r_2,-r_1],\quad (-r_1,r_1),\quad [r_1,r_2],\quad (r_2,1),
$$
and then fix the relative constants by requiring $f^\infty$ to be continuous at the four matching points. This avoids ever having to track an integral of the form $\int_0^w H(v)\chi_S(v)\,dv$, whose lower limit $0$ need not lie in the same connected component as $w$.

On the three subintervals where $H=0$ (the two tails and the central band), the ODE is the pure one solved above, so on each of them $f^\infty$ retains the unperturbed Beta form \eqref{eq:f_beta}--\eqref{eq:ab_beta_op}, with a constant that may differ from one subinterval to the next: we denote them by $C_-(x)$ on $(-1,-r_2)$, $C_0(x)$ on $(-r_1,r_1)$, $C_+(x)$ on $(r_2,1)$.

On $S$,  $H(v) = (w_d -v)\chi_{\{v \in S\}}$, since the opinion dynamics is influenced also by the action of the media. If we still assume $\alpha=1/2$, the stationary distribution becomes
\begin{align}
f^\infty(w,x)
= \frac{C(x)}{1-w^2}
\exp\!\left(\frac{2}{\lambda}
\int_0^w
\frac{\dfrac{1}{\tau_0}(m_W(x)-v\rho_W(x))
+\dfrac{k^2}{\nu\tau_1}H(v)}
{1-v^2}\,dv\right).
\end{align}
Splitting the integral and using the result of \eqref{eq:f_beta},
we can factor out the Beta contribution to obtain
\begin{align*}
f^\infty(w,x)
= C(x)\,(1-w)^{a(x)-1}(1+w)^{b(x)-1}
\exp\!\left(\frac{2k^2}{\lambda\nu\tau_1}
\int_0^w \frac{H(v)}{1-v^2}\,dv\right),
\label{eq:f_H}
\end{align*}
where $a(x)$ and $b(x)$ are given by \eqref{eq:ab_beta_op}. The additional
exponential factor represents a deformation of the Beta distribution induced by the media.

We compute the decomposition
\begin{equation*}
\frac{w_d - v}{1-v^2}
= \frac{w_d-v}{(1+v)(1-v)}
= \frac{A}{1+v}+\frac{B}{1-v},
\end{equation*}
and we get
\begin{equation*}
A = \frac{w_d+1}{2}, \qquad B = \frac{w_d-1}{2}.
\label{eq:AB}
\end{equation*}
Therefore, 
\begin{align*}
\int_{[0,w]\cap S} \frac{H(v)}{1-v^2}\,dv
&= \frac{w_d+1}{2}\log(1+w) - \frac{w_d-1}{2}\log(1-w).
\end{align*}
The exponential factor thus becomes
\begin{align*}
\exp\!\left(\frac{2k^2}{\lambda\nu\tau_1}
\int_0^w\frac{H(v)}{1-v^2}\,dv\right)
= (1+w)^{\frac{k^2(w_d+1)}{\lambda\nu\tau_1}}
(1-w)^{-\frac{k^2(w_d-1)}{\lambda\nu\tau_1}},
\qquad w\in S.
\end{align*}
Combining this with the pure contribution from \eqref{eq:f_beta}, the ODE on each component of $S$ is solved by the stationary distribution on $S$:
\begin{align*}
f^\infty(w,x)
= \Tilde{C}(x)\,(1-w)^{\tilde{a}(x)-1}(1+w)^{\tilde{b}(x)-1},
\qquad w\in S,
\label{eq:f_beta_H}
\end{align*}
where
\begin{equation*}
\tilde{a}(x)
= \frac{\rho_W(x)-m_W(x)}{\lambda\tau_0}
- \frac{k^2(w_d-1)}{\lambda\nu\tau_1},
\qquad
\tilde{b}(x)
= \frac{\rho_W(x)+m_W(x)}{\lambda\tau_0}
+ \frac{k^2(w_d+1)}{\lambda\nu\tau_1},
\label{eq:ab_shifted}
\end{equation*}
a Beta form with shifted
parameters and a priori independent multiplicative constant on each component of $S$: $\tilde C_-(x)$ on $[-r_2,-r_1]$, $\tilde C_+(x)$ on $[r_1,r_2]$..
Since $w_d\in[-1,1]$, we note that $w_d-1\leq 0$ and $w_d+1\geq 0$, so
that both corrections to $a(x)$ and $b(x)$ are non-negative. The media
therefore shifts both Beta parameters upward, with the net asymmetry
controlled by the sign of $w_d$: a positive target opinion $w_d>0$
increases $\tilde{b}$ more than $\tilde{a}$, pushing the distribution toward $w=+1$, as expected on physical grounds, a negative target opinion $w_d <0$ has the opposite effect.

The five local solutions are pinned together, up to the single overall constant $C(x)$, by continuity at $\partial S$. Fixing $C_0(x):=C(x)$ and matching at $w=r_1$ and $w=r_2$ gives
$$
\tilde C_+(x) = C(x)\left({1-r_1}\right)^{a(x)-\tilde a(x)}\!\!(1+r_1)^{b(x)-\tilde b(x)},
\qquad
C_+(x) = \tilde C_+(x)\,(1-r_2)^{\tilde a(x)-a(x)}(1+r_2)^{\tilde b(x)-b(x)},
$$
by symmetry and matching at $w=-r_1$ and $w=-r_2$ gives:
$$
\tilde C_-(x) = C(x)\,(1+r_1)^{a(x)-\tilde a(x)}(1-r_1)^{b(x)-\tilde b(x)},
\qquad
C_-(x) = \tilde C_-(x)\,(1+r_2)^{\tilde a(x)-a(x)}(1-r_2)^{\tilde b(x)-b(x)}.
$$
Finally, $C(x)$ is fixed by the normalisation $\int_{-1}^1 f^\infty(w,x)\,dw=1$, each of the five terms being an elementary Beta integral.\\
Overall, the stationary distribution is a piecewise Beta distribution,
\begin{equation}
f^\infty(w,x) =
\begin{cases}
C_\pm(x)\,(1-w)^{a(x)-1}(1+w)^{b(x)-1}, & w \in (r_2,1)\ \text{or}\ (-1,-r_2),\\[4pt]
C_0(x)\,(1-w)^{a(x)-1}(1+w)^{b(x)-1}, & w \in (-r_1,r_1),\\[4pt]
\tilde C_\pm(x)\,(1-w)^{\tilde{a}(x)-1}(1+w)^{\tilde{b}(x)-1}, & w \in [r_1,r_2]\ \text{or}\ [-r_2,-r_1],
\end{cases}
\end{equation}
with shifted parameters $\tilde{a}(x), \tilde{b}(x)$ on $S$ and unperturbed
parameters $a(x), b(x)$ outside $S$, and with the five constants $C_-(x),C_0(x),C_+(x),\tilde C_-(x),\tilde C_+(x)$ determined, up to the single free constant $C(x)$, by the matching relations above. By construction, the two branches meet continuously
at $\partial S$, though with a discontinuity in the first derivative, reflecting
the sharp boundary of the media's influence region.

\begin{lemma}
Let the stationary density $f^\infty(w,x)$ be defined on the domain $w \in [-1,1]$ as:
\begin{equation}
f^{\infty}(w,x) = \frac{C(x)}{D^2(w)} \exp\Bigg( \frac{2}{\lambda} \int_0^w \frac{F(v,x)}{D^2(v)} \, dv \Bigg)
\label{eq:f_infinity_drift}
\end{equation}
with the effective drift $F(v,x)$  given by:
\begin{equation}
F(v,x) \coloneqq \frac{1}{\tau_0}(m_W(x)-v\rho_W(x)) + \frac{k^2}{\nu\tau_1}(w_d-v)\chi_{\{v \in S\}}
\label{eq:drift}
\end{equation}
where $D^2(w) = (1-w^2)^{2\alpha}$ with $\alpha > 0$ represents the state-dependent diffusion coefficient.
For any fixed $x \in \Omega$, the distribution $f^\infty(\cdot,x)$ is integrable in $L^1(-1,1)$ provided that $w_d \in (-1,1)$ and the system is in a non-trivial state ($m_W < \rho_W$).
\end{lemma}

\begin{proof}
In order to study the integrability of $f^{\infty}$ we need to analyse what happens at the boundaries $w = \pm 1$, where the diffusion coefficient $D^2(w)$ vanishes (see \cite{During:strongleaders}).
We analyse the behaviour near $w=1$; the case $w=-1$ follows by symmetry.
Let $\beta \in (1-\epsilon,1)$ be a value sufficiently close to the right boundary and consider the integral:
\begin{equation*}
\mathcal{I} = \int_\beta^1 f^\infty(w,x) \, dw.
\end{equation*}
We evaluate the sign of the drift $F(v,x)$  defined in Equation \eqref{eq:drift} as $v \to 1$.

Let us consider the non-trivial situation in which the social interaction term satisfies $m_W(x) \leq \rho_W(x)$.
We notice that since the function $r(v)= m_W(x)-v\rho_W(x)$ is continuous and $r(1)= \rho_W(x)-m_W(x)>0 $ in a non trivial case. Then there exists a $\beta <1$ and $\chi >0$ s.t. $v-m \geq \chi \quad \forall v \in [\beta,1]$. The same observation can be done for $h(v)= v-w_d$ if $h(1)= 1- w_d >0$. Hence, choosing $w_d <1$ there exists a $\zeta <1$ and $K > 0$ s.t.\ $v - w_d \geq K \quad \forall v \in [\zeta,1]$.

Let us define $\epsilon \coloneqq \max\{\beta, \zeta\}<1$. 
By continuity, there exists a neighbourhood $[\epsilon, 1]$ and a constant $\chi(x) > 0$ such that $F(v,x) \leq -(\chi + K)$ for all $v \in [\epsilon, 1]$.

Regarding the support of the media $S$:
\begin{itemize}
    \item If $S = [-1,1]$, i.e. the action of the media involves all the agents, the term $(w_d-v)$ is strictly negative near $v=1$ since $w_d < 1$.
    \item If $S \subsetneq [-1,1]$, the media influence does not appear in our integral because $\chi_S(v)$ vanishes next to the boundaries, and the proof follows from \cite{During:strongleaders}.
\end{itemize}
We split the integration interval in $[0, \epsilon] \cup [\epsilon,w]$:
\begin{align*}
    \int_0^w \frac{v\rho_W(x)- m_W(x)}{D^2(v)} dv +  \int_0^w \frac{v-w_d}{D^2(v)} dv =& \int_0^\epsilon \frac{v\rho_W(x)- m_W(x)}{D^2(v)} dv + \int_\epsilon^w \frac{v\rho_W(x)- m_W(x)}{D^2(v)} dv \\& + \int_0^\epsilon \frac{v-w_d}{D^2(v)} dv  + \int_\epsilon^w \frac{v-w_d}{D^2(v)} dv .
\end{align*}
We notice that the first and the third integrals in $[0,\epsilon]$ are positive and finite since the stochastic coefficient is bounded for values different from the extreme values. Hence:
\begin{align*}
\int_0^w \frac{v\rho_W(x)- m_W(x)}{D^2(v)} dv + \int_0^w \frac{v-w_d}{D^2(v)} dv &\geq \int_\epsilon^w \frac{v\rho_W(x)- m_W(x)}{D^2(v)} dv + \int_\epsilon^w \frac{v-w_d}{D^2(v)} dv\\ & \geq \chi \int_\epsilon^w \frac{1}{D^2(v)} dv + K \int_\epsilon^w \frac{1}{D^2(v)} dv\\
& = (\chi + K) \int_\epsilon^w  \frac{1}{D^2(v)} dv.
\end{align*}
Due to the monotonicity of the exponential function, we conclude that:
$$ \exp\Big(- \int_0^w \frac{v\rho_W(x)- m_W(x)}{D^2(v)} dv - \int_0^w \frac{v-w_d}{D^2(v)} dv\Big)  \leq  \exp\Big(- (\chi+K)  \int_\beta^w \frac{dv}{D^2(v)} \Big) .$$
If we choose for simplicity all parameters equal to one, the integral $\mathcal{I}$ is thus bounded by:
\begin{equation*}
\mathcal{I} \leq C \int_\beta^1 \frac{1}{D^2(w)} \exp\Bigg( -c \int_\beta^w \frac{dv}{D^2(v)} \Bigg) \, dw,
\end{equation*}
with $c = 2(\chi+K) > 0$.\\
We introduce the change of variables:
\begin{equation*}
u(w) \coloneqq \int_\beta^w \frac{dv}{D^2(v)}, \quad du = \frac{dw}{D^2(w)}.
\end{equation*}
The behaviour of the integral depends on the limit $u(1) = \int_\beta^1 (1-v^2)^{-2\alpha} \, dv$:
\begin{itemize}
    \item if $\alpha < 1/2$: $u(1) = U < \infty$. The integral $\int_0^U e^{-cu} du < \infty$.
    \item if $\alpha \geq 1/2$: $u(1) = \infty$. However, the exponential decay dominates:
\begin{align*}
\int_{\beta}^{1}\frac{1}{D^2(w)} \exp\Big(- (\chi + K)\int_\beta^w \frac{dv}{D^2(v)} \Big) dw =& \int_0^\infty \exp(-(\chi +K) u) du\\& = -\frac{1}{\chi +K}\int_0^{\infty} - (\chi +K) \exp(-(\chi +K) u)du \\ & = \frac{1}{\chi +K}
\end{align*}
\end{itemize}
In both regimes of $\alpha$, we conclude that $f^\infty(\cdot,x) \in L^1(-1,1)$, ensuring the existence of a normalised stationary state.
\end{proof}

\section{Exponential convergence in Fourier metric on a graphon}
\label{sec:equilibration}
In this section we study the exponential convergence of solutions of the Boltzmann-type equation \eqref{eq:boltz} in suitable Fourier-based metrics \cite{gabetta1995metrics,carrillo2007contractive}.
In these papers and related works, the distance,
\begin{equation}
d_s(f_1,f_2)
=
\sup_{\xi\neq 0}
\frac{|\hat f_1(\xi)-\hat f_2(\xi)|}{|\xi|^s},
\qquad s>0,
\label{eq:distance}
\end{equation}
where

$$\hat f(\xi,t)=\int_{-1}^1 e^{-i\xi w} f(w,t)\,dw.$$
is considered which is finite provided that the statistical moments of the two probability measures coincide up to the order $[s]$, the integer part of $s$, which means $[s]=s$ if $s \in \mathbb{R} \setminus \mathbb{N}$ and $[s]=s-1$ otherwise. The metric $d_s$ is equivalent (with explicit constants) to the $1$-Wasserstein distance from optimal transport theory.
 A detailed review of Fourier metrics can be found in \cite{carrillo2007contractive}.

We notice that, compared with the classical convergence results, our unknown $f$ depends on an additional variable $x$, hence we need a distance
that resolves the opinion distribution \textit{at each graph
position $x$ separately}. We therefore introduce, for each fixed
$x\in[0,1]$, the Fourier transform
\begin{equation}
\hat f(\xi,x,t) := \int_{-1}^{1} e^{-i\xi w} f(w,x,t)\,dw,
\label{eq:fourier_transform_def}
\end{equation}
and the associated metric
\begin{equation}
D_s(f_1,f_2)(t) :=
\sup_{x\in[0,1]}\;
\sup_{\xi\neq 0}
\frac{|\hat f_1(\xi,x,t)-\hat f_2(\xi,x,t)|}{|\xi|^s}.
\label{eq:Ds_def}
\end{equation}
Since $|e^{-i\xi w}|=1$ for all $w,\xi$, we always have the pointwise
bound
\begin{equation*}
|\hat f(\xi,x,t)| \le \int_{-1}^1 f(w,x,t)\,dw = \rho(x,t) \le 1,
\label{eq:fhat_bound}
\end{equation*}
where $\rho(x,t)$ is the local mass at $x$, and the last inequality
holds because $f(\cdot,x,t)$ is a probability
density in $w$.

By construction, $D_s$ dominates the classical marginal distance.
Denote by
\begin{equation*}
d_s(f_1,f_2)(t) := \sup_{\xi\neq0} \frac{|\hat f_1(\xi,t)-\hat f_2(\xi,t)|}{|\xi|^s},
\qquad
\hat f(\xi,t) := \int_{-1}^1 e^{-i\xi w} f(w,t)\, dw,
\label{eq:ds_def}
\end{equation*}
the analogue of \eqref{eq:Ds_def} for the $x$-marginal opinion
densities $f_1(w,t)=\int_0^1 f_1(w,x,t)\,dx$ and
$f_2(w,t)=\int_0^1 f_2(w,x,t)\,dx$. Then
\begin{equation}
d_s(f_1,f_2)(t) \;\le\; D_s(f_1,f_2)(t), \qquad \forall t\ge0.
\label{eq:ds_leq_Ds}
\end{equation}

We work with the graphon Boltzmann-type gain operator $\mathcal
Q_0^+$ associated with the opinion dynamics, defined in weak form, for a test function $\varphi(w,x)$, by
\begin{equation}
\big(\mathcal Q_0^+[f](\cdot,\cdot,t),\varphi\big)
= \int_0^1\!\!\int_0^1\!\!\int_{-1}^1\!\!\int_{-1}^1
W(x,y)\,\big\langle \varphi(w',x)\big\rangle\,
f(w,x,t)\,f(w_*,y,t)\, dw\,dw_*\,dy\,dx,
\label{eq:Q0_def}
\end{equation}
where $\langle\cdot\rangle$ denotes the average over the stochastic
part of the interaction rule, $w'$ is the post-interaction opinion, and
$\mathcal Q_1^+$ denotes the analogous (graphon-independent) gain
operator describing the interaction with the media.

The full local equation for $\hat f(\xi,x,t)$ reads
\begin{equation}
\partial_t \hat f(\xi,x,t)
= \frac{1}{\tau_0}\Big[\widehat{\mathcal Q_0^+}[f](\xi,x,t) - \rho_W(x,t)\,\hat f(\xi,x,t)\Big]
+ \frac{1}{\tau_1}\Big[\widehat{\mathcal Q_1^+}[f](\xi,x,t) - \hat f(\xi,x,t)\Big],
\label{eq:full_local_eq}
\end{equation}

where $\rho_W(x,t)$ is the graphon-weighted local mass defined in \eqref{eq:mass_media_weighted}, and the media loss rate is
uniform in $x$ (equal to $1$), since the media background carries no
spatial structure.

\begin{assumption}
\label{ass:setting}
We consider two solutions $f_1(t),f_2(t)$ of \eqref{eq:full_local_eq},
evolving on the same graphon $W(x,y)$, such that:
\begin{enumerate}
\item[(i)] the local mass is pointwise conserved and equal for the
two solutions,
\begin{equation*}
\rho_1(x,t)=\rho_2(x,t) =: \rho(x,t) = \rho(x,0), \qquad \forall x,t,
\label{eq:mass_conservation}
\end{equation*}
which holds because agents do not change their position in the graph; only
their opinion $w$ evolves in time, so the mass at each fixed $x$ is
conserved (and consequently $\rho_{W,1}(x,t)=\rho_{W,2}(x,t)=:\rho_W(x,t)$
for all $x,t$);
\item[(ii)] the total mean opinion of the system is conserved so both solutions meet the requirements of the Lemma \ref{lem:mean_control};
\item[(iii)] the stochastic component does not feature in the opinion model, i.e. $\eta=\eta_* = 0$.
\end{enumerate}
\end{assumption}

\begin{theorem}
\label{thm:fourier}
Let $f_1(t),f_2(t)$ be two solutions of \eqref{eq:full_local_eq}
satisfying Assumption~\ref{ass:setting}. Let $s\in(1,2]$, with
$\gamma\in(0,1/2)$, and assume $D_s(f_1(0),f_2(0))<\infty$. Then
\begin{equation}
D_s(f_1(t),f_2(t)) \le D_s(f_1(0),f_2(0))\,
\exp\!\left(\int_0^t \kappa_W(\tau)\,d\tau\right), \qquad \forall t\ge0,
\label{eq:main_estimate}
\end{equation}
where
\begin{equation}
\kappa_W(t) :=
\frac{\underline\rho_W(t)}{\tau_0}\Big[\big(\gamma^s+(1-\gamma)^s\big)-1\Big]
+ \frac{1}{\tau_1}\left[\left(\frac{\lambda}{\lambda+\delta^2}\right)^{\!s}-1\right],
\qquad
\underline\rho_W(t):=\inf_{x\in[0,1]}\rho_W(x,t).
\label{eq:kappa_final}
\end{equation}
Moreover, $\kappa_W(t)<0$ for every $t\ge0$, independently of the
graphon $W$. Therefore, $D_s(f_1(t),f_2(t))\to0$ as $t\to+\infty$, and
every solution converges to the stationary state $f^\infty$
uniformly in $x$.
\end{theorem}

\begin{proof}
 Choose test functions in
\eqref{eq:Q0_def} of the form $\varphi(w,x) =
\psi(w)\,\chi_\varepsilon(x-x_0)$, where $\chi_\varepsilon$ is an
approximate identity concentrated around $x=x_0$, so that
$\chi_\varepsilon(x-x_0)\to\delta(x-x_0)$ as $\varepsilon\to0$.
Substituting into \eqref{eq:Q0_def} and letting $\varepsilon\to0$
collapses the outer integral in $x$ onto $x=x_0$, giving
\begin{equation}
\big(\mathcal Q_0^+[f](w,x_0,t),\psi\big)
= \int_0^1\!\!\int_{-1}^1 \int_{-1}^1
W(x_0,y)\,\big\langle \psi(w')\big\rangle\,
f(w,x_0,t)\,f(w_*,y,t)\, dw\,dw_*\,dy,
\label{eq:local_weak}
\end{equation}
in which $x_0$ is now a fixed parameter, no longer an integration
variable. Since the graph position $x$ never changes in time (only
the opinion $w$ evolves), the equation \eqref{eq:local_weak} represents the
evolution equation for the sub-population of agents located at
$x_0$, and it is this equation, resolved in $x$, that must be
analysed in order to keep track of how the graphon's local structure
affects the convergence rate.\\
Take
$\psi(w)=e^{-i\xi w}$ in \eqref{eq:local_weak} and define
$\widehat{\mathcal Q_0^+}[f](\xi,x_0,t) := (\mathcal
Q_0^+[f](w,x_0,t),\,e^{-i\xi(w)})$. Since $\eta=0$ then
gives $e^{-i\xi w'} =
e^{-i\xi(1-\gamma)w}e^{-i\xi\gamma w_*}$, so that
\begin{align}
\widehat{\mathcal Q_0^+}[f](\xi,x_0,t)
&= \int_0^1 W(x_0,y)
\left[\int_{-1}^1 e^{-i\xi(1-\gamma)w} f(w,x_0,t)\, dw\right]
\left[\int_{-1}^1 e^{-i\xi\gamma w_*} f(w_*,y,t)\, dw_*\right] dy
\notag\\
&= \hat f\big((1-\gamma)\xi,x_0,t\big)
\int_0^1 W(x,y)\, \hat f(\gamma\xi,y,t)\, dy.
\label{eq:Q0hat_graphon}
\end{align}\\
We notice that for $W\equiv1$, when interactions are not influenced by graph position, this reduces to
$\rho_W(x,t)=\int_0^1\rho(y,t)\,dy=1$ and to the classical relation
$\widehat{\mathcal Q_0^+}[f](\xi,t)=\hat f((1-\gamma)\xi,t)\,\hat
f(\gamma\xi,t)$ of the all-to-all case.

 Let
$h(\xi,x,t):=\hat f_1(\xi,x,t)-\hat f_2(\xi,x,t)$. Since $f_1,f_2$
share the same graphon-weighted local mass $\rho_W(x,t)$ by
Assumption~\ref{ass:setting}(i), subtracting \eqref{eq:Q0hat_graphon}
written for $f_1$ and for $f_2$, adding and subtracting the mixed
product $\hat f_1((1-\gamma)\xi,x_0,t)\int_0^1 W(x_0,y)\hat
f_2(\gamma\xi,y,t)\,dy$, and grouping terms gives
\begin{align*}
\widehat{\mathcal Q_0^+}[f_1]-\widehat{\mathcal Q_0^+}[f_2]
={}& \hat f_1\big((1-\gamma)\xi,x_0,t\big)
\int_0^1 W(x_0,y)\, h(\gamma\xi,y,t)\, dy \notag\\
&+ \left[\int_0^1 W(x_0,y)\, \hat f_2(\gamma\xi,y,t)\, dy\right]
h\big((1-\gamma)\xi,x_0,t\big).
\end{align*}
By the triangle inequality and the pointwise bound
\eqref{eq:fhat_bound},
\begin{align}
\big|\widehat{\mathcal Q_0^+}[f_1]-\widehat{\mathcal Q_0^+}[f_2]\big|(\xi,x_0,t)
&\le \rho_W(x_0,t)\, \sup_y\big|h(\gamma\xi,y,t)\big|
+ \rho_W(x_0,t)\, \sup_y\big|h\big((1-\gamma)\xi,y,t\big)\big|.
\label{eq:triangle}
\end{align}
Dividing \eqref{eq:triangle} by $|\xi|^s$ and using the scaling
identity 
\begin{equation}
    \frac{|h(a\xi,y,t)|}{|\xi|^s} = |a|^s\,\frac{|h(a\xi,y,t)|}{|a\xi|^s} \le
|a|^s\, D_s(t)
\label{eq:rescaling_fourier}
\end{equation}
with $D_s(t):=D_s(f_1,f_2)(t)$, applied with
$a=\gamma$ and $a=1-\gamma$, yields
\begin{equation}
\frac{\big|\widehat{\mathcal Q_0^+}[f_1]-\widehat{\mathcal Q_0^+}[f_2]\big|(\xi,x_0,t)}{|\xi|^s}
\le \rho_W(x_0,t)\big(\gamma^s+(1-\gamma)^s\big)\, D_s(t).
\label{eq:Q0_final_bound}
\end{equation}
The media operator $\mathcal Q_1^+$ has no graphon dependence, so
its Fourier estimate is unchanged from the classical (all-to-all)
setting. Proceeding as in \cite{pareschi2013BOOK} for the case of opinion interaction with the effect of an external background, we can consider our media a background that we will call $M$ dependent on the target opinion $w_d$:
\begin{equation} \big(\mathcal Q_1^+(f,M),\psi\big) = \int_{[-1,1]^2} f(w) M(w_d) \varphi(w'') dw dw_d.
\label{eq:scalar_prod_backg}
\end{equation}
Choosing $\varphi(w)=e^{-i\xi w}$ in \eqref{eq:scalar_prod_backg}, denoting by $\widehat{f_S}$ the Fourier transform restricted to $\xi\! \!\in\! \! S$ and considering $M(w_d)=\delta(w_d - w_d^*)$, we get
\begin{equation*}
\widehat{\mathcal{Q}_1}(f)(\xi)=\widehat M\left(\frac{\delta^2}{\lambda+\delta^2}\xi\right)
\widehat{f_S}\left(\frac{\lambda}{\lambda+\delta^2}\xi\right).
\end{equation*}
Evaluating the operator for the two solutions, we have
\begin{equation*}
\left|
\widehat{\mathcal{Q}_1^+}(f_1)(\xi) - \widehat{\mathcal{Q}_1^+}(f_2)(\xi)
\right| = \left|
\widehat{f_{1,S}}\left(\frac{\lambda}{\lambda+\delta^2}\xi\right)
-\widehat{f_{2,S}}\left(\frac{\lambda}{\lambda+\delta^2}\xi\right)
\right|.
\end{equation*}
since, it is easy to prove that: 
\begin{equation*}
    \left|\widehat M\left(\frac{\delta^2}{\lambda+\delta^2}\xi\right)\right| = 1.
\end{equation*}
Setting $\xi' = \frac{\lambda}{\lambda+\delta^2}\xi$ and exploiting the property in \eqref{eq:rescaling_fourier}, we obtain
\begin{equation*}
\frac{
\left|
\widehat{\mathcal{Q}_1^+}(f_1)(\xi) - \widehat{\mathcal{Q}_1^+}(f_2)(\xi)
\right|
}{|\xi|^s}
=\left(\frac{\lambda}{\lambda+\delta^2}\right)^s
\frac{
\left|
\widehat{f_{1,S}}(\xi') - \widehat{f_{2,S}}(\xi')
\right|
}{|\xi'|^s}.
\end{equation*}
Taking the supremum, we deduce
\begin{equation*}
d_s({\mathcal{Q}}_1^+(f_1),{\mathcal{Q}}_1^+(f_2))
\le
\left(\frac{\lambda}{\lambda+\delta^2}\right)^s
d_s(f_{1,S},f_{2,S}).
\end{equation*}
Since 
\begin{equation*}
d_s(f_{1,S},f_{2,S}) \le d_s(f_{1},f_{2}),
\end{equation*}
we end up with:
\begin{equation}
\frac{\big|\widehat{\mathcal Q_1^+}[f_1]-\widehat{\mathcal Q_1^+}[f_2]\big|(\xi,x,t)}{|\xi|^s}
\le \left(\frac{\lambda}{\lambda+\delta^2}\right)^{\!s} D_s(t).
\label{eq:Q1_final_bound}
\end{equation}
 Subtracting
\eqref{eq:full_local_eq} written for $f_1$ and $f_2$ and using
$\rho_{W,1}(x,t)=\rho_{W,2}(x,t)=\rho_W(x,t)$ gives the equation
satisfied by $h$,
\begin{equation}
\partial_t h(\xi,x_0,t)
= \frac{1}{\tau_0}\Big[\big(\widehat{\mathcal Q_0^+}[f_1]-\widehat{\mathcal Q_0^+}[f_2]\big) - \rho_W(x_0,t)\, h(\xi,x_0,t)\Big]
+ \frac{1}{\tau_1}\Big[\big(\widehat{\mathcal Q_1^+}[f_1]-\widehat{\mathcal Q_1^+}[f_2]\big) - h(\xi,x_0,t)\Big].
\label{eq:h_local_eq}
\end{equation}
Dividing \eqref{eq:h_local_eq} by $|\xi|^s$ and applying
\eqref{eq:Q0_final_bound}--\eqref{eq:Q1_final_bound} gives, for every $x_0\in[0,1]$ and $\xi\neq0$,
\begin{equation}
\partial_t \frac{|h(\xi,x_0,t)|}{|\xi|^s}
\le
\left[\frac{\rho_W(x_0,t)}{\tau_0}\big(\gamma^s+(1-\gamma)^s\big) + \frac{1}{\tau_1}\left(\frac{\lambda}{\lambda+\delta^2}\right)^{\!s}\right] D_s(t)
\;-\; \left[\frac{\rho_W(x_0,t)}{\tau_0} + \frac{1}{\tau_1}\right] \frac{|h(\xi,x_0,t)|}{|\xi|^s}.
\label{eq:pointwise_ineq}
\end{equation}
Equation \eqref{eq:pointwise_ineq} holds for every $x_0,\xi$; we differentiate
$D_s(t)=\sup_{x,\xi}|h(\xi,x,t)|/|\xi|^s$. By the classical argument
for differentiating a supremum of a family of functions satisfying a
uniform differential inequality, evaluated at the point
$(\xi^\star,x^\star)$ where the supremum is attained at
time $t$, one has $|h(\xi^\star,x^\star,t)|/|\xi^\star|^s = D_s(t)$,
so that \eqref{eq:pointwise_ineq} yields:
\begin{equation}
\partial_t D_s(t) \le \sup_{x_0\in[0,1]}
\left\{\frac{\rho_W(x_0,t)}{\tau_0}\Big[\big(\gamma^s+(1-\gamma)^s\big)-1\Big]\right\}
D_s(t)
+ \frac{1}{\tau_1}\left[\left(\frac{\lambda}{\lambda+\delta^2}\right)^{\!s}-1\right] D_s(t).
\label{eq:kappa_sup}
\end{equation}
Since $\gamma\in(0,1/2)$ and
$s\in(1,2]$, one has $\gamma^s+(1-\gamma)^s < 1$, so the bracket
$\big[(\gamma^s+(1-\gamma)^s)-1\big]$ is a negative constant,
independent of $x$. Multiplying a negative constant by
$\rho_W(x,t)\in[\underline\rho_W(t),\overline\rho_W(t)]$, the product
is largest (i.e., least negative) precisely when $\rho_W(x,t)$ is
smallest, so that
\begin{equation}
\sup_{x_0\in[0,1]}\Big\{\rho_W(x_0,t)\big[(\gamma^s+(1-\gamma)^s)-1\big]\Big\}
= \underline\rho_W(t)\,\big[(\gamma^s+(1-\gamma)^s)-1\big].
\label{eq:sign_argument}
\end{equation}
Substituting \eqref{eq:sign_argument} into \eqref{eq:kappa_sup} gives
$\partial_t D_s(t) \le \kappa_W(t)\, D_s(t)$ with $\kappa_W(t)$ as in
\eqref{eq:kappa_final}, and Gr\"onwall's lemma yields
\eqref{eq:main_estimate}.

Both terms in
\eqref{eq:kappa_final} are non-positive: the first because
$\underline\rho_W(t)\ge0$ and $(\gamma^s+(1-\gamma)^s)-1<0$; the
second because $\lambda/(\lambda+\delta^2)<1$ whenever $\delta>0$,
hence $\left(\lambda/(\lambda+\delta^2)\right)^s-1<0$ for every
$s>0$. The second term alone is strictly negative and does not
depend on $W$, so $\kappa_W(t)<0$ always, even in the degenerate
case $\underline\rho_W(t)=0$ (vanishing minimal degree somewhere in
the network). Consequently $D_s(f_1(t),f_2(t))\to0$ as
$t\to+\infty$, and since $D_s$ dominates $d_s$ by
\eqref{eq:ds_leq_Ds}, this implies
$\sup_{x\in[0,1]} d_s\big(f_1(\cdot,x,t),f_2(\cdot,x,t)\big)\to0$,
i.e., convergence uniform in $x$.
\end{proof}

\begin{corollary}
\label{cor:stationary}
Let $f(t)$ solve \eqref{eq:full_local_eq} on the graphon $W$ satisfy the condition of Theorem \ref{thm:fourier}, and let $f^\infty(w,x)$ be the stationary
distribution. Then
\begin{equation*}
D_s\big(f(t),f^\infty\big) \le D_s\big(f(0),f^\infty\big)\,
\exp\!\left(\int_0^t \kappa_W(\tau)\, d\tau\right), \qquad \forall t\ge0,
\end{equation*}
and
\begin{equation*}
\lim_{t\to+\infty}\; \sup_{x_0\in[0,1]}\, d_s\big(f(\cdot,x_0,t),f^\infty(\cdot,x_0)\big) = 0,
\end{equation*}
i.e., every solution converges to the stationary state $f^\infty$
uniformly in $x$.
\end{corollary}

\begin{proof}
Since $f^\infty$ solves the same equation on the same $W$, with the
same local mass and $\partial_t f^\infty=0$, the pair $(f,f^\infty)$
satisfies Assumption~\ref{ass:setting}. The claim then follows by
applying Theorem~\ref{thm:fourier} with $f_1=f$, $f_2\equiv f^\infty$.
\end{proof}

\input{numerical_section}

\section{Conclusions}

We proposed a kinetic model of opinion formation on graphon-modulated
networks under selective media influence, in which the media acts, through
a constrained model predictive control available in closed form, only on
the agents in the sphere of legitimate controversy. For the resulting
Boltzmann-type equation we studied the moment evolution, derived the quasi-invariant limit Fokker-Planck-type equation, characterised its stationary state, of piecewise Beta-type with parameters shifted by the media inside $S$, and proved
exponential convergence to equilibrium in the Fourier metric.

The numerical experiments clarify the interplay between the interaction
structure and the control. On synthetic graphon networks, consensus at the
target opinion requires either global compromise or global media reach:
when both are localised, the extreme agents are shielded from the control
by the confidence threshold and polarisation persists, while on densely
connected networks the binary interactions relay the media effect, so that
a control acting on a smaller controversy set suffices. On two real-world,
single-issue Twitter networks, initialised from the Barber\'a ideal points
so that both the interaction structure and the initial opinions are read
from the connection structure alone, without any textual analysis, the
dynamics reproduces the same qualitative picture: mediated agents
accumulate at the target inside $S$, moderate agents undergo media-induced
sign changes, and the deviance clusters at $w=\pm1$ persist. Most notably,
the asymptotic opinion distribution of both real networks is recovered, up
to Fourier distances of order $10^{-3}$, by a one-parameter Gaussian
graphon with a single calibrated bandwidth, common to the two datasets.
This indicates that the macroscopic dynamics depends on the interaction
structure essentially through its mean intensity, and makes the graphon a
cheap surrogate, with $\mathcal{O}(N)$ storage and freely scalable
population size, for opinion dynamics on real networks.

Several directions remain open. For instance, the graphon is static in
time, while real networks rewire on time scales comparable to those of
opinion formation, and the comparison with data concerns the asymptotic
state only, since the datasets provide a single temporal snapshot; a
validation against observed temporal evolutions would provide a more
stringent test. A further direction in the AI realm stands at the
interface with physically driven generative models: variational
auto-encoders and text encoders could embed textual content into a
low-dimensional opinion space on which the present dynamics acts as a
physically principled prior, and samples steered by the controlled
dynamics in the latent space could be decoded back to text, in analogy
with back-mapping in coarse-grained molecular modelling and with latent
diffusion models, yielding opinion-conditioned sampling and generation. Another interesting direction could be to consider the effect of two competing media aiming to drive agents towards two different target opinions on overlapping subsets of the sphere of controversy.

\section*{Acknowledgements} 
\begin{small}
M.F.\ and A.L.\ worked under the auspices of the Italian National Group of Mathematical Physics (GNFM) of INdAM. A.L.\ was supported by the Project Piano Nazionale di Ripresa e Resilienza -- Next Generation EU (PNRR-NGEU) from the Italian Ministry of University and Research (MUR) under Grant DM~117/2023. A.L.\ expresses his gratitude to Prof.\ Lamberto Rondoni for valuable discussions and support. Computational resources were provided by \texttt{hpc@polito}, a project of Academic Computing within the Department of Control and Computer Engineering at the Politecnico di Torino (\url{http://www.hpc.polito.it}). M.F.\ is grateful to Warwick Mathematics Institute, where a part of this research has been carried out, for the kind hospitality and thanks Prof.\ Andrea Tosin and Dr.\ Nadia Loy for their invaluable supervision, constant encouragement, and many fruitful discussions that greatly contributed to this work.
\end{small}

\appendix
\section{Appendix}
\subsection{Boundedness of post-interaction opinions}
\label{app:bounded}

We discuss sufficient conditions to ensure that the binary interaction rules preserve the admissible opinion domain $[-1,1]$.
Consider the microscopic interactions
\begin{align}
w' &= w + \gamma P(w,w_*)(w_*-w) + \eta D(w),\\
w_*' &= w_* + \gamma P(w_*,w)(w-w_*) + \eta_* D(w_*),
\end{align}
where $w,w_*\in[-1,1]$ denote the pre--interaction opinions.
We assume $0\le P(w,w_*)\le 1$,
$0<\gamma<1/2,$
and choose the diffusion coefficient of the form
$$D(w)=(1-w^2)^\alpha.$$
The compromise term can be rewritten as
$$w+\gamma P(w,w_*)(w_*-w)
=
(1-\gamma P(w,w_*))w
+
\gamma P(w,w_*)w_*.$$
Hence, 
the deterministic interaction part is a convex combination of $w$ and $w_*$ and therefore belongs to $[-1,1]$ whenever $w,w_*\in[-1,1]$.

Consequently, the possible loss of invariance is entirely due to the random fluctuation term.
In particular, choosing
$$D(w)=(1-w^2)^\alpha,
\qquad 0<\alpha<1,$$
we can rewrite
$$D(w)
=
(1-|w|)^\alpha (1+|w|)^\alpha.$$
Hence, the stochastic contribution satisfies
$$|\eta|D(w)
=
|\eta|(1-|w|)^\alpha (1+|w|)^\alpha.$$
To preserve the admissible interval uniformly with respect to $w\in[-1,1]$, it is sufficient to require
$$|\eta|D(w)\le (1-\gamma)(1-|w|).$$
Substituting the expression of $D(w)$ gives
$$|\eta|
\le
(1-\gamma)
\frac{(1-|w|)^{1-\alpha}}{(1+|w|)^\alpha}.$$
We now analyse the behaviour of the right--hand side near the boundary
$|w|\to1$.

If $\alpha<1$, then
$$(1-|w|)^{1-\alpha}\to0
\qquad \text{as } |w|\to1,$$
and therefore no strictly positive uniform bound on $|\eta|$ can exist independently of $w$. In this regime, the diffusion term vanishes too slowly near the boundaries, and support preservation may fail unless additional corrections are introduced.

If $\alpha=1$ (see \cite{toscani2006kinetic}), then
$$\frac{(1-|w|)^{1-\alpha}}{(1+|w|)^\alpha}
=
\frac{1}{1+|w|},$$
so that
$$|\eta|
\le
\frac{1-\gamma}{1+|w|}$$
and, using $1+|w|\le2$, a sufficient uniform condition is
$$|\eta|
\le
\frac{1-\gamma}{2}.$$

If $\alpha>1$, then
$$(1-|w|)^{1-\alpha}\to+\infty
\qquad \text{as } |w|\to1.$$
 Hence, the diffusion term degenerates faster near the boundaries. The preservation condition becomes even less restrictive, and any sufficiently small bounded noise amplitude is admissible.

An analogous argument can be made for $\eta_*$.
Therefore, $\alpha\ge1$
is precisely the condition ensuring the existence of a non-trivial uniform bound on the random fluctuations guaranteeing $w',w_*'\in[-1,1]$
for all binary interactions.

\printbibliography

\end{document}

%% file: numerical_section.tex
\section{Numerical Experiments}
\label{sec:numerics}

In this section we illustrate the behaviour of the model through direct Monte Carlo simulations of the microscopic system \eqref{post-int-op1}--\eqref{post-int-media}, and we quantify the extent to which a one-parameter graphon can reproduce the opinion dynamics generated by a real-world social network. This section is organised in four parts. For completeness, we recall the direct Monte Carlo scheme of Nanbu--Babovsky type in Subsection~\ref{subsec:dmc}. In Subsection~\ref{subsec:synthetic} we consider synthetic data with non-informative initial conditions and a Gaussian graphon with randomly assigned graph positions, and we investigate the qualitative evolution of the opinion density and its dependence on the control and interaction parameters. In Subsection~\ref{subsec:realdata} we instantiate the model on two real-world single-issue social networks, using the empirical network together with initial opinions estimated from data, and we examine the effect of the selective media control on the sphere of controversy $S$. Finally, in Subsection~\ref{subsec:surrogate} we show that the asymptotic opinion distribution produced by the real-world social network can be recovered by a synthetic Gaussian graphon whose single bandwidth parameter $\sigma_W$ is suitably tuned, even though the graphon is real-valued in $(0,1)$ while the empirical adjacency matrix is binary. We stress that the recovery is an equivalence at the level of the opinion marginal $f(w,t)$, not a statement of microscopic identity between the two interaction structures.

Distances between distributions are measured in the Fourier metric $d_s$ given in \eqref{eq:distance}, the metric in which the convergence Theorem~\ref{thm:fourier} is stated, which quantifies the weak convergence of probability measures in the same spirit as the Wasserstein distances of optimal transport theory.

\subsection{Direct Monte Carlo scheme}
\label{subsec:dmc}

We approximate the Boltzmann-type dynamics \eqref{eq:boltz} by a direct Monte Carlo scheme of Nanbu--Babovsky type \cite{pareschi2013BOOK}. At each step of length $\Delta t$ the agents are randomly matched in pairs; an agent is activated for binary opinion exchange with probability $\Delta t/\tau_0$ and for interaction with the media with probability $\Delta t/\tau_1$. The opinion update follows \eqref{post-int-op1}, with weight given by the graphon $W(x_i,x_j)$ in the synthetic case and by the adjacency matrix entry $A_{ij}$ in the real-world social network case; the media update applies the optimal control of Lemma~\ref{lem:wellposed_kkt}, that is, the projection of the unconstrained response onto the admissible set $S$. The diffusion is modulated by $D(w)=(1-w^2)^\alpha$ with the opinion-dependent noise truncation of Appendix~\ref{app:bounded}, ensuring $w_i(t)\in[-1,1]$ for every $t$.

\begin{algorithm}[h]
\caption{Direct Monte Carlo scheme for the controlled opinion dynamics}
\label{alg:dmc}
\begin{algorithmic}[1]
\STATE \textbf{Input:} agents $N$; horizon $T$, step $\Delta t$, $n_{\max}=\lfloor T/\Delta t\rfloor$; rates $\tau_0,\tau_1$; compromise $\gamma$; noise $\sigma$; diffusion exponent $\alpha$; media parameters $\delta,\tilde\lambda,w_d$; controversy set $S=[-r_2,-r_1]\cup[r_1,r_2]$; interaction structure (graphon $W$ with graph positions $\{x_i\}$, or adjacency matrix $A$); initial opinions $\{w_i(0)\}_{i=1}^N$.
\STATE \textbf{Output:} empirical opinion distribution $\{w_i(t)\}$ at the sampled times.
\FOR{$n=1$ \TO $n_{\max}$}
  \STATE draw a random permutation $\pi$ of $\{1,\dots,N\}$
  \FOR{each agent $i$ in parallel}
    \STATE $j\leftarrow\pi(i)$;\quad $W_{ij}\leftarrow W(x_i,x_j)$ (synthetic) \textbf{ or } $A_{ij}$ (real)
    \STATE $\Theta_i\sim\mathrm{Bernoulli}(\Delta t/\tau_0)$;\quad $P_{ij}\leftarrow \mathbf 1_{\{|w_i-w_j|\le r\}}$ or $1$
    \STATE sample $\eta_i$, $\mathbb E\eta_i=0$, $\mathrm{Var}\,\eta_i=\sigma^2$, truncated to $|\eta_i|\le(1-\gamma)(1-|w_i|)^{1-\alpha}(1+|w_i|)^{-\alpha}$
    \STATE $w_i^{\mathrm{free}}\leftarrow w_i+\Theta_i\big[\gamma\,W_{ij}P_{ij}(w_j-w_i)+\eta_i(1-w_i^2)^\alpha\big]$
    \STATE $w_m^{\mathrm{free}}\leftarrow (\tilde\lambda\,w_i^{\mathrm{free}}+\delta^2 w_d)/(\tilde\lambda+\delta^2)$
    \STATE $w_m^{*}\leftarrow \sgn(w_m^{\mathrm{free}})\,\mathrm{Proj}_{[r_1,r_2]}(|w_m^{\mathrm{free}}|)$
    \STATE $\Theta_i^{m}\sim\mathrm{Bernoulli}(\Delta t/\tau_1)$;\quad $\chi_S\leftarrow \mathbf 1_{\{w_i^{\mathrm{free}}\in S\}}$
    \STATE $w_i\leftarrow (1-\Theta_i^{m}\chi_S)\,w_i^{\mathrm{free}}+\Theta_i^{m}\chi_S\,w_m^{*}$
  \ENDFOR
  \STATE if $n$ is a sampling time, record the histogram of $\{w_i\}$
\ENDFOR
\end{algorithmic}
\end{algorithm}
\noindent Convergence to the stationary regime is identified by monitoring the plateau of the opinion variance and the stabilisation of $d_s$ between successive sampling windows; all asymptotic quantities are computed on the post-plateau window.

\subsection{Synthetic data}
\label{subsec:synthetic}

We first consider non-informative initial conditions, in which both the graph positions and the opinions are sampled uniformly,
\[
x_i\sim\mathcal U([0,1]),\qquad w_i(0)\sim\mathcal U([-1,1]),
\]
so that no structure is imposed on the initial configuration and the patterns that emerge are intrinsic to the dynamics rather than inherited from the data. The interaction is modulated by the Gaussian graphon
\begin{equation}
W(x,y)=\exp\!\left(-\frac{|x-y|^2}{\sigma_W^2}\right),
\label{eq:gaussian-graphon}
\end{equation}
evaluated at the random positions $\{x_i\}$. Since the graph position $x$ and the opinion $w$ are independent, the bandwidth $\sigma_W$ does not encode opinion homophily: it controls the mean interaction strength $\rho_W=\mathbb E_{x,y}[W(x,y)]$, a monotone function of $\sigma_W$, and it is through this scalar that the graphon enters the macroscopic dynamics. Despite the absence of initial structure, the system develops non-trivial patterns driven by the interplay of binary compromise, stochastic fluctuations and selective media influence, as shown in Figure~\ref{fig:synthetic-evolution}. The evolution typically exhibits three phases: an initial transient, an intermediate clustering regime, and a final stationary state.

Figure~\ref{fig:synthetic-evolution} also isolates the joint role of the compromise function $P$ and of the controversy set $S$. The three panels share all parameters and differ only in this pair. In panel~(a), bounded-confidence compromise combined with a strictly localised $S$ leaves the population polarised: agents whose opinions fall inside $S$ are drawn toward the target $w_d$, while the extreme agents near $w=\pm1$ neither interact with the moderate bulk, because of the confidence threshold, nor are reached by the media, because they lie outside $S$, and therefore persist at the boundaries. In panel~(b), complete interactions ($P\equiv1$) allow compromise to propagate across the whole interval, so that even the extreme agents are progressively pulled in and the distribution collapses onto a single cluster near $w_d$. In panel~(c), bounded confidence is retained but the media acts on the entire interval ($S=[-1,1]$), which again removes the protected extremes and yields consensus. Consensus thus requires either global compromise or global media reach; when both are localised, polarisation is preserved. This equivalence has a practical implication: on a densely connected network, where compromise already propagates opinions across the whole population, the media can achieve the same consensus while acting on a smaller controversy set $S$, since the binary interactions relay its effect from the mediated agents to the rest.

\begin{figure}[ht]
\centering
\includegraphics[width=\textwidth]{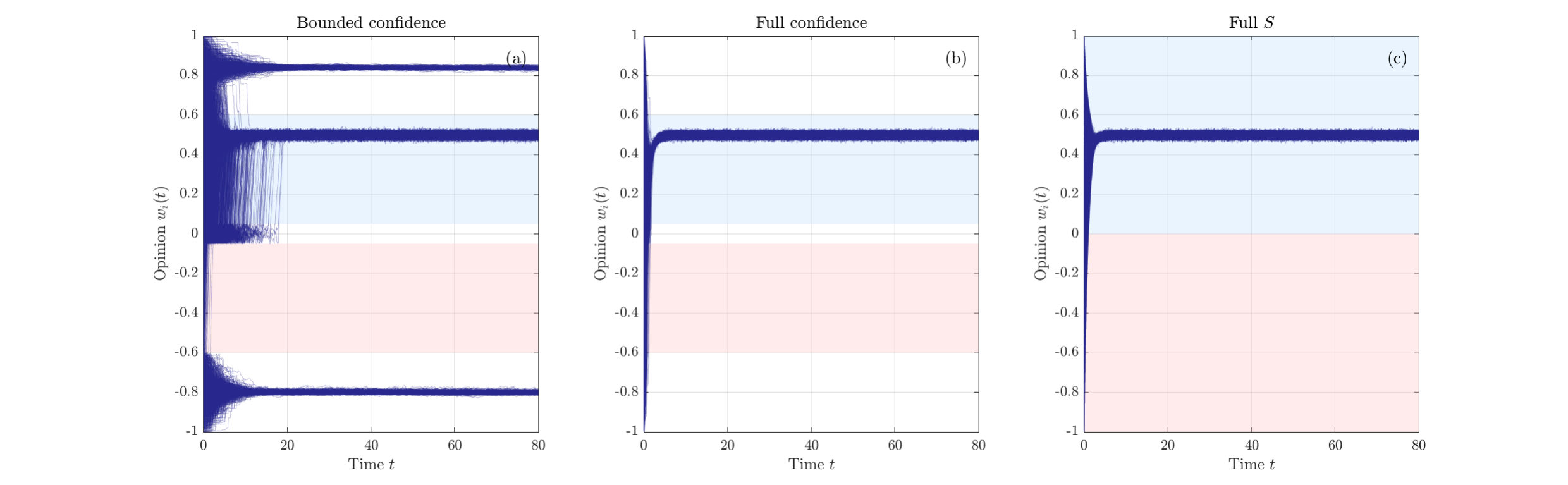}
\caption{Trajectories of the agents' opinions $w_i(t)$ for synthetic data with uniform initial conditions and the Gaussian graphon \eqref{eq:gaussian-graphon}. The shaded bands mark the controversy set $S$. The three panels share $N=4000$, $\Delta t=0.1$, $T=80$, $\tau_0=\tau_1=0.1$, $\gamma=0.4$, $\sigma=0.01$, $\alpha=2$, $\sigma_W=0.35$, confidence radius $r=0.25$, and media block $\delta=0.15$, $\tilde\lambda=0.2$, $w_d=0.5$. They differ only in the pair $(P,S)$: (a) bounded confidence with $S=[-0.6,-0.05]\cup[0.05,0.6]$; (b) complete interactions $P\equiv1$ with the same $S$; (c) bounded confidence with $S=[-1,1]$. Consensus at $w_d$ emerges in (b) and (c), whereas the localised configuration (a) preserves the polar clusters.}
\label{fig:synthetic-evolution}
\end{figure}

\subsubsection{Effect of the control and interaction parameters}
\label{subsec:param-study}
Within the synthetic setting we now isolate the role of the controversy set $S$, the media strength $\delta$, and the compromise function $P$ at the level of the asymptotic distribution. Unless otherwise specified the runs use the parameters of Figure~\ref{fig:synthetic-evolution}.

\paragraph{Role of the controversy set $S$.}
The set $S=[-r_2,-r_1]\cup[r_1,r_2]$ acts as a filter that selects which agents are exposed to the media, thereby introducing heterogeneity into the dynamics. Figure~\ref{fig:media_effect_S_bounded} reports the asymptotic distributions for three nested choices of $S$, from the full interval to a narrow band. Enlarging $S$ increases the fraction of the population reached by the control and produces a stronger concentration of opinions around the target $w_d$: the peak at $w_d$ is highest for $S=[-1,1]$ and progressively lower as $S$ shrinks, while the mass retained away from the target correspondingly grows. For the narrowest set only moderately opinionated agents are influenced, and the alignment toward $w_d$ is weakest.

\begin{figure}[ht]
\centering
\includegraphics[width=0.8\textwidth]{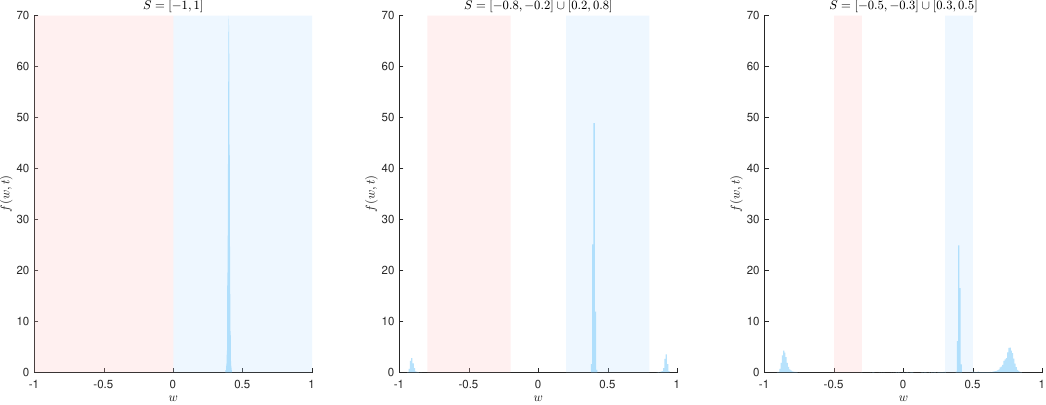}
\caption{Opinion distributions $f(w,t)$ at $t=1000$ for three nested choices of the set $S=[-r_2,-r_1]\cup[r_1,r_2]$, from the full interval $S=[-1,1]$ to $S=[-0.8,-0.2]\cup[0.2,0.8]$ and $S=[-0.5,-0.3]\cup[0.3,0.5]$, under the bounded-confidence opinion dynamics with confidence radius $r=0.2$. Parameters: $\gamma=0.4$, $\sigma=0.05$, $\sigma_W=0.2$, $\tilde\lambda=0.05$, $w_d=0.4$, $\delta=0.5$, $\alpha=2$. The peak at $w_d$ is highest for $S=[-1,1]$ and progressively lower as $S$ shrinks, leading to other clusters close to the extreme opinions.}
\label{fig:media_effect_S_bounded}
\end{figure}
\paragraph{Effect of the media strength.}
Figures~\ref{fig:media_effect_full} and~\ref{fig:media_effect_bounded} isolate the role of the media strength $\delta$, for complete and bounded-confidence interactions, respectively, comparing the opinion distribution at a short time $t=5$ with the long-time distribution at $t=1000$. Under complete interactions ($P\equiv 1$) shown in Figure~\ref{fig:media_effect_full}, a weak media leaves the transient dynamics essentially unaffected, and at $t=5$ the population aggregates near the mean of the initial distribution, as in classical consensus models. As $\delta$ increases, the agents inside $S$ concentrate at the target $w_d$ already in the transient, while the mass outside $S$ drifts toward the target without crossing $\partial S$; at $t=1000$ the distribution has collapsed onto a single narrow peak at $w_d$, whose height grows and whose width shrinks with $\delta$, showing that a stronger media produces a faster and sharper convergence of the whole population. The peak is not a Dirac mass, its finite width being sustained by the stochastic term. Under bounded confidence, shown in Figure~\ref{fig:media_effect_bounded}, the picture is different: for every value of $\delta$ the long-time distribution retains, alongside the dominant peak at $w_d$, two residual clusters near the extremes, populated by agents that the confidence threshold disconnects from the moderate bulk and that lie outside the reach of the media, in agreement with the mechanism of Figure~\ref{fig:synthetic-evolution}. Increasing $\delta$ sharpens the central peak but does not remove the extreme clusters. We refer to the distributions at $t=1000$ as numerical steady states, since residual lumps of mass adjacent to $S$ are still slowly attracted toward $w_d$ on longer time scales.
\begin{figure}[ht]
\centering
\includegraphics[width=0.8\textwidth]{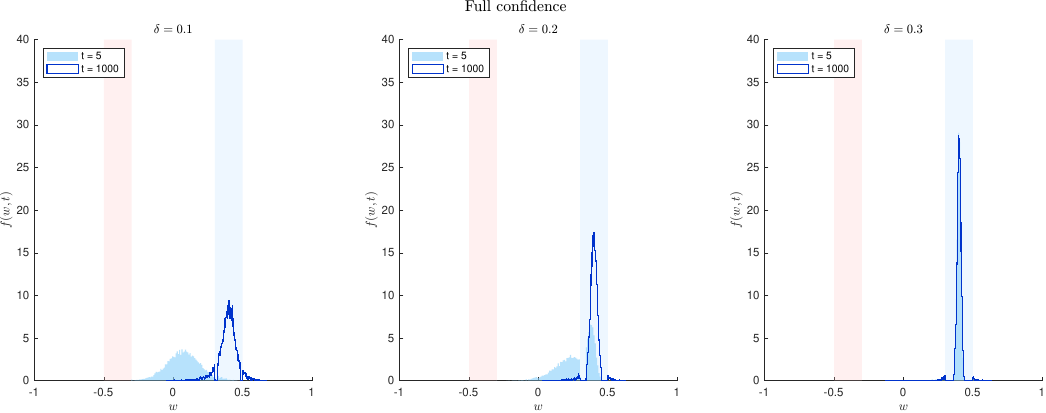}
\caption{Opinion distributions $f(w,t)$ at $t=5$ and $t=1000$ for media strengths $\delta\in\{0.1,0.2,0.3\}$ under complete interactions ($P\equiv1$). Remaining parameters: $\gamma=0.4$, $\sigma=0.05$, $\sigma_W=0.2$, $\tilde\lambda=0.05$, $w_d=0.4$, $S=[-0.5,-0.3]\cup[0.3,0.5]$, $\alpha=2$. A stronger media yields a faster and sharper concentration of the whole population at $w_d$.}
\label{fig:media_effect_full}
\end{figure}

\begin{figure}[ht]
\centering
\includegraphics[width=0.8\textwidth]{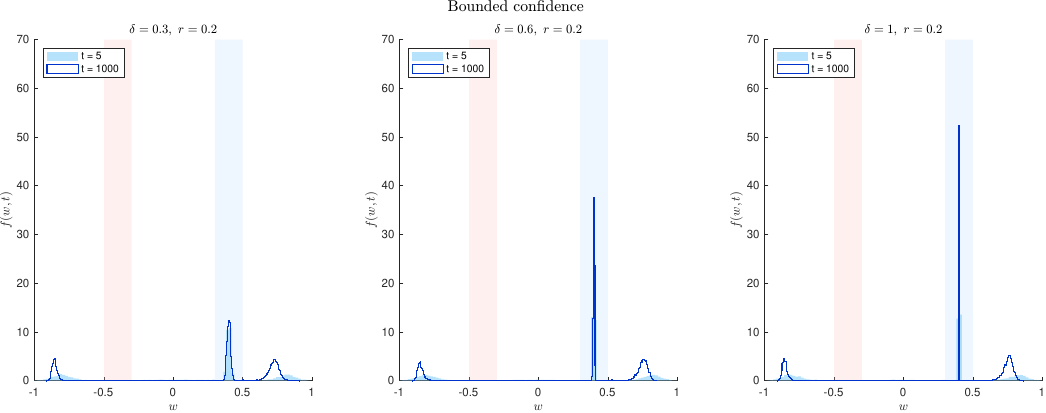}
\caption{Opinion distributions $f(w,t)$ at $t=5$ and $t=1000$ for media strengths $\delta\in\{0.3,0.6,1\}$ under bounded-confidence interactions with confidence radius $r=0.2$. Remaining parameters as in Figure~\ref{fig:media_effect_full}. The central peak at $w_d$ sharpens as $\delta$ grows, while the clusters near the extremes persist, shielded from both compromise and control.}
\label{fig:media_effect_bounded}
\end{figure}

\paragraph{Media-induced sign change.}
The control can reverse the sign of the opinion of agents that are sufficiently close to zero. Neglecting for simplicity the stochastic and the binary-interaction contributions, the media update in the unconstrained regime reads
\[
w''_{\mathrm{free}}=\frac{\tilde\lambda\,w+\delta^2 w_d}{\tilde\lambda+\delta^2},
\]
with target $w_d>0$. For an agent with $w<0$, a change of sign occurs whenever $w''_{\mathrm{free}}>0$, that is $\tilde\lambda\,w+\delta^2 w_d>0$, which gives
\[
|w|<\frac{\delta^2 w_d}{\tilde\lambda}.
\]
Agents with moderately negative opinions are thus shifted to positive opinions by the media, while those with $|w|>\delta^2 w_d/\tilde\lambda$ retain their sign. For instance, with $\delta=0.3$, $\tilde\lambda=0.2$ and $w_d=0.5$ the threshold is $\delta^2 w_d/\tilde\lambda=0.225$, so agents with initial opinion in $(-0.225,0)$ are driven toward positive opinions, whereas more radical agents remain negative. Examples how agents change the sign of their opinion can be observed in Figure \ref{fig:synthetic-evolution}.

\subsection{Real-world network data}
\label{subsec:realdata}

\paragraph{Materials.}
To probe the model on empirical interaction structures we require networks built around a single political issue, so that a one-dimensional opinion variable is meaningful. We use two such single-issue Twitter networks, concerning gun control and the Affordable Care Act (Obamacare), drawn from the polarised-debate collection of Garimella et al.\ \cite{garimella2018www,garimella2018political}, publicly available at \href{https://rutgers.box.com/s/9lavww3g0uyajkcvxwu66sxtauncq589}{rutgers.box.com}. For each issue the nodes are the users active in the debate and the edges record their interactions. The two datasets differ in size but share the same structural character; their main statistics are collected in Table~\ref{tab:network-stats}. Both networks are sparse, with fewer than $4\%$ of the possible connections present, yet the mean degree is of the order of $10^2$ interactions per user. Connectivity is strongly heterogeneous and right-skewed: the maximum degrees exceed the medians by one to two orders of magnitude, and in the Obamacare network the median out-degree is zero, so that a small core of highly active users coexists with a largely passive majority, as is typical of online debate. The out-degree tails are heavier than the in-degree ones, reflecting the asymmetry of the following relation.

\begin{table}[h]
\centering
\begin{tabular}{lrr}
\toprule
 & Gun control & Obamacare \\
\midrule
Nodes $N$                    & $3543$    & $7902$ \\
Directed edges               & $493\,872$ & $1\,056\,143$ \\
Edge density $\mathrm{nnz}(A)/N^2$ & $3.93\times10^{-2}$ & $1.69\times10^{-2}$ \\
Mean degree                  & $139.4$   & $133.7$ \\
Median out-degree / in-degree & $84$ / $77$ & $0$ / $71$ \\
Max out-degree / in-degree    & $984$ / $926$ & $2191$ / $1182$ \\
$95$th perc.\ out-degree / in-degree & $466$ / $479$ & $715$ / $498$ \\
\bottomrule
\end{tabular}
\caption{Structural statistics of the two single-issue networks. Degrees refer to the directed adjacency matrix $A$, with $A_{ij}=1$ if user $i$ follows user $j$.}
\label{tab:network-stats}
\end{table}
%

\paragraph{Initial opinions from Barber\'a ideal points.}
We initialise the opinion of each user from the political ideal-point
estimate of Barber\'a \cite{barbera2015}. The method places users in a
latent ideological space through a Bayesian spatial-following model: a
user's position is inferred from the set of political accounts they
follow, under the assumption that following is more likely between
ideologically proximate accounts. The estimator is introduced in
\cite{barbera2015} against established measures of ideology, and the
resulting one-dimensional score is a well-established proxy for political
position. We rescale the estimated ideal points linearly to the opinion
interval $[-1,1]$ and use them as the initial datum
$\{w_i(0)\}_{i=1}^N$, retaining the network topology exactly.

The natural alternative would be to recover the initial opinions from the
text of the tweets by sentiment analysis, and it is worth explaining why
the ideal points are preferable in our setting. A sentiment classifier
returns the affective valence of a message, whereas the opinion variable
$w$ is a position on an issue axis, and the two need not agree: on a
polarised issue, users on opposite sides typically both produce messages
of negative valence directed at the opposing camp, so that the sign of
the sentiment does not identify the sign of the opinion. The difficulty
is compounded by the ironic and sarcastic register characteristic of
polarised online debate, on which text classifiers are known to degrade;
the ideal point sidesteps the issue entirely, as it never reads the text.
Moreover, the ideal point is by construction a single scalar per user,
matching the agent state $w_i(t)$ of our model, whereas a text-based
estimate would require classifying and aggregating a very large corpus of
tweets per user, with the attendant computational cost and per-message
noise. Finally, the ideal point is inferred from the same object that our
model takes as primitive, namely the structure of the connections, so
that opinions and interaction structure are read from a single source.

We remark that the ideal point measures a general ideological position
rather than a stance on the specific issue. This identification is
appropriate here, since both debates are strongly aligned with the
partisan axis, but it would be weaker on an issue that cuts across party
lines. The empirical initial distributions, shown in
Figure~\ref{fig:initial-distributions}, exhibit the bimodal, polarised
shape characteristic of single-issue debate, with the two lobes
reflecting the opposing camps.

\begin{figure}[ht]
\centering
\includegraphics[width=\textwidth]{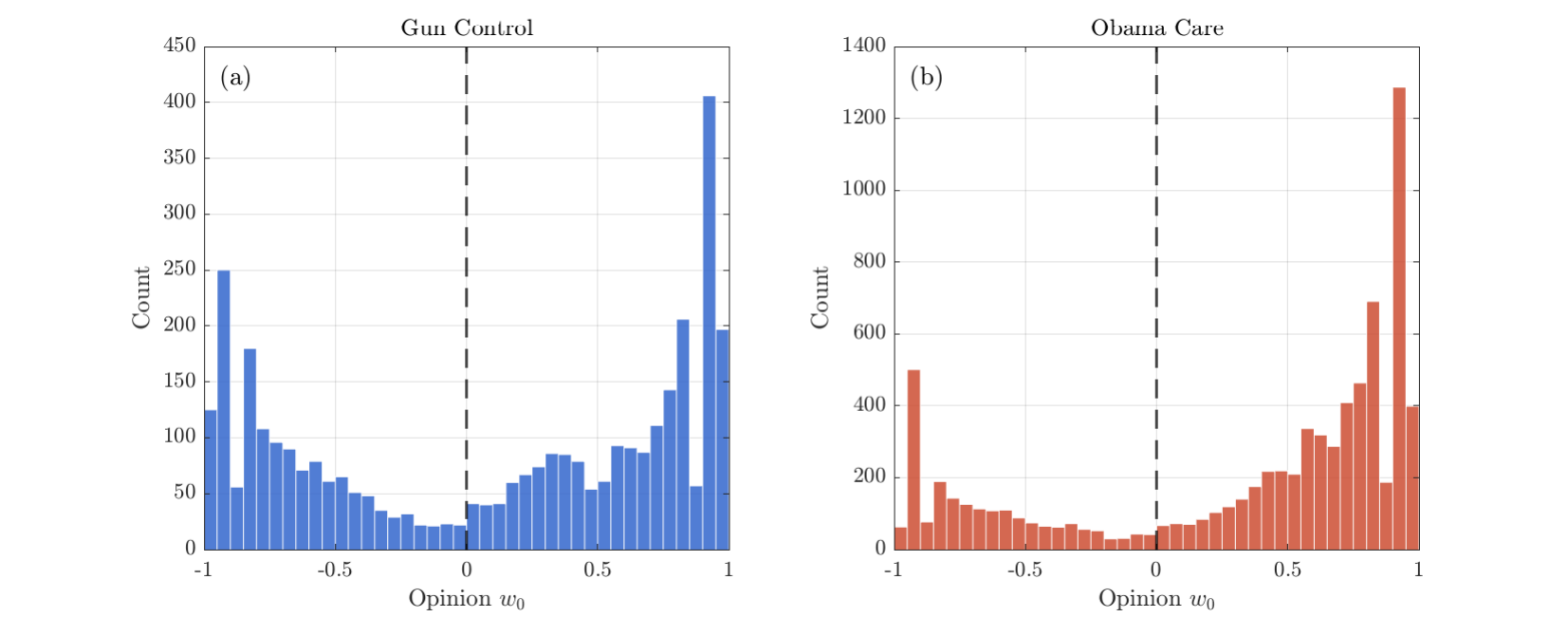}
\caption{Initial opinion distributions obtained from the normalised Barber\'a ideal points for the two single-issue datasets, (a) gun control and (b) Obamacare. Both are strongly bimodal, with mass concentrated near the two opinion poles and a depleted moderate centre.}
\label{fig:initial-distributions}
\end{figure}

\paragraph{Graph reconstruction.}

The interaction operator is assembled from the raw interaction list as a sparse binary adjacency matrix $A$, with $A_{ij}\in\{0,1\}$ recording the presence of an interaction between users $i$ and $j$, restricted to the component on which Barber\'a scores are available. The matrix is used directly as provided by the data, so that the directed structure of the empirical interactions is preserved. The binary matching of Algorithm~\ref{alg:dmc} then samples interaction partners and reads the corresponding entries $A_{ij}$. We stress that in this setting the matrix $A$ simply replaces the graphon
evaluations $W(x_i,x_j)$ in Algorithm~\ref{alg:dmc} as an empirical
interaction kernel, and is not itself a graphon \cite{franceschi2022}. Being
directed, it is not required to satisfy the symmetry assumption under which the
analytical results of the previous sections are established, and the
real-network experiments should be read as a test of the model beyond
the assumptions of the theory.

\paragraph{Controlled evolution on the real networks.}
We evolve the dynamics on each network from the Barber\'a initial datum, including the media control on the sphere of controversy $S$. The resulting trajectories are reported in Figure~\ref{fig:real-evolution}. The media acts selectively on the agents whose opinions lie in $S$, driving them toward the target $w_d$, while agents in the spheres of consensus and deviance evolve under the binary opinion dynamics alone. Two features are visible in both datasets. First, a dense band forms inside $S$ near the upper edge of the controversy region, where the mediated agents accumulate around $w_d$. Second, the trajectories display sporadic transitions from the negative pole to the positive band, the discrete-time signature of the media-induced sign change analysed above, while the most extreme agents near $w=\pm1$, lying in the sphere of deviance, retain their position throughout. The two networks behave qualitatively alike, despite their different sizes and connectivity.

\begin{figure}[ht]
\centering
\includegraphics[width=\textwidth]{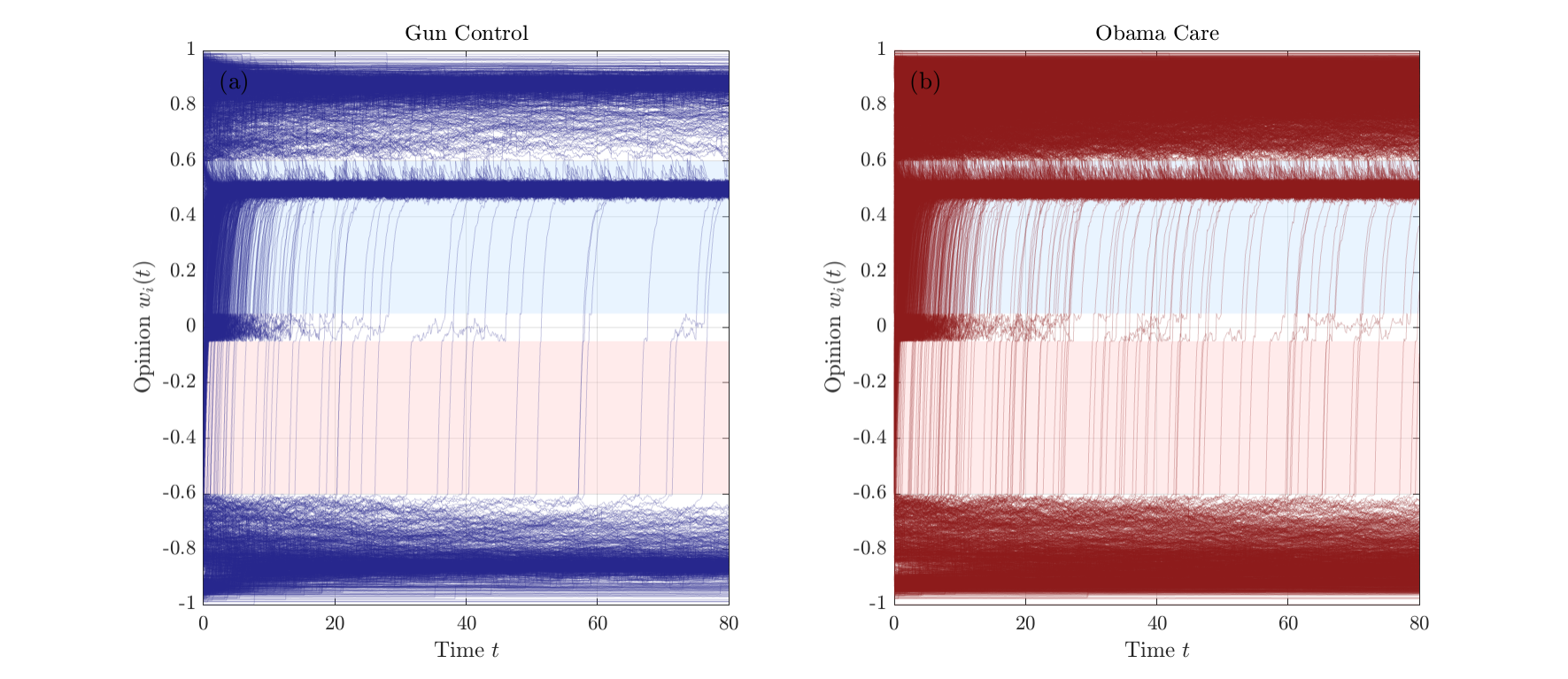}
\caption{Trajectories of the agents' opinions on the two real networks, (a) gun control and (b) Obamacare, initialised from the Barber\'a scores and subject to bounded-confidence compromise and the media control on $S=[-0.6,-0.05]\cup[0.05,0.6]$. Shared parameters: $\Delta t=0.1$, $T=80$, $\tau_0=\tau_1=0.1$, $\gamma=0.4$, $\sigma=0.01$, $\alpha=2$, confidence radius $r=0.25$, $\delta=0.15$, $\tilde\lambda=0.2$, $w_d=0.5$. The dense band inside $S$ marks the mediated agents accumulating near $w_d$; the vertical streaks are sign-change transitions; the persistent clusters near $w=\pm1$ are the deviance agents, untouched by the control.}
\label{fig:real-evolution}
\end{figure}

\subsection{Graphon surrogate: recovering the asymptotic distribution}
\label{subsec:surrogate}

We now turn to the central experiment. We run the dynamics from the same empirical initial datum $f_0$ given by the normalised Barber\'a scores, with identical opinion-space and media parameters, changing only the interaction structure: once with the real adjacency matrix $A$, and once with the Gaussian graphon \eqref{eq:gaussian-graphon} evaluated at random graph positions $x_i\sim\mathcal U([0,1])$. We then ask whether the asymptotic opinion distribution of the real network can be reproduced by the synthetic graphon for a suitable choice of the single bandwidth $\sigma_W$.

We emphasise the structural difference between the two interaction operators. The empirical adjacency matrix is binary, $A_{ij}\in\{0,1\}$, encoding the discrete presence or absence of a connection, whereas the graphon \eqref{eq:gaussian-graphon} is real-valued in $(0,1)$ and assigns a smoothly varying interaction intensity to every pair. The fact that the graphon can nonetheless match the asymptotic opinion marginal of the real network indicates that the macroscopic dynamics is insensitive to the fine structure of the interaction kernel and depends, at leading order, only on its mean intensity $\rho_W$.

To recover the real dynamics we sweep $\sigma_W$ and select, for each dataset, the value $\sigma_W^\star$ that minimises the Fourier distance $d_s$ between the asymptotic empirical distribution of the synthetic graphon and that of the real network. The minimiser is essentially the same for the two datasets, $\sigma_W^\star=0.4445$, which reflects the similar mean interaction strength of the two networks and confirms that a single scalar suffices to match the dynamics. The coincidence of the two calibrated values should be read at the resolution of the search, since the bandwidth is swept on a logarithmically spaced grid and $\sigma_W^\star$ is identified up to the local grid spacing. Within this tolerance the two minimisers are indistinguishable, consistently with the fact that the graphon enters the macroscopic dynamics only through the mean interaction strength $\rho_W(\sigma_W)$ and the two networks have similar mean connectivity. Figure~\ref{fig:surrogate-evolution} compares the full time evolution of the densities, and Figure~\ref{fig:surrogate-histograms} the resulting asymptotic distributions, at the calibrated bandwidth.

\begin{figure}[ht]
\centering
\includegraphics[width=\textwidth]{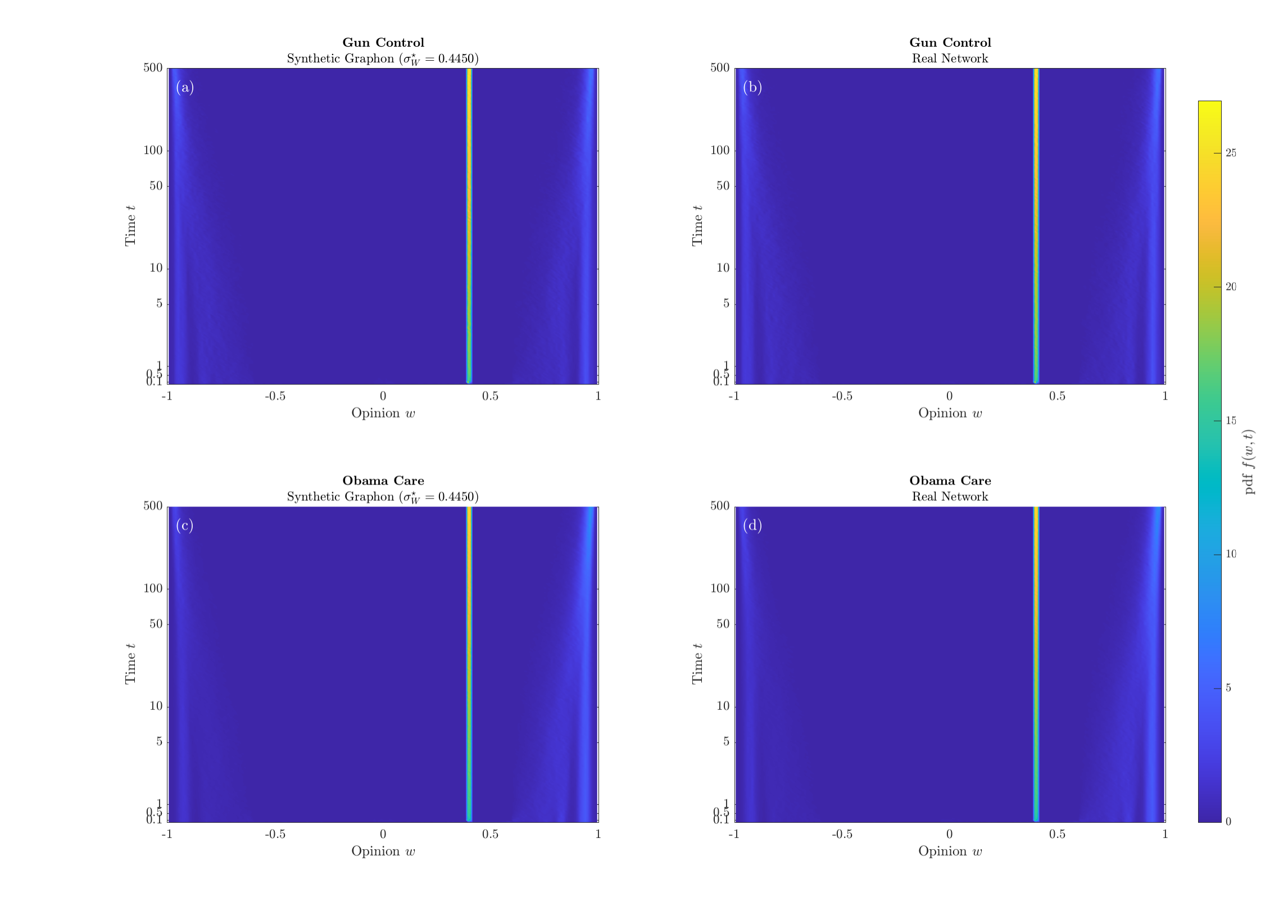}
\caption{Evolution of the opinion density $f(w,t)$ on a log-spaced time axis, comparing the calibrated synthetic graphon ($\sigma_W^\star=0.445$) with the real network: gun control (top, a synthetic and b real) and Obamacare (bottom, c synthetic and d real), all initialised from the Barber\'a scores under identical opinion-space and media parameters. The bright vertical line marks the concentration of the mediated agents at the target opinion inside $S$, and the residual mass near $w=\pm1$ the deviance agents. The synthetic and real columns are visually indistinguishable throughout the evolution, not only at the final time.}
\label{fig:surrogate-evolution}
\end{figure}

\begin{figure}[ht]
\centering
\includegraphics[width=0.92\textwidth]{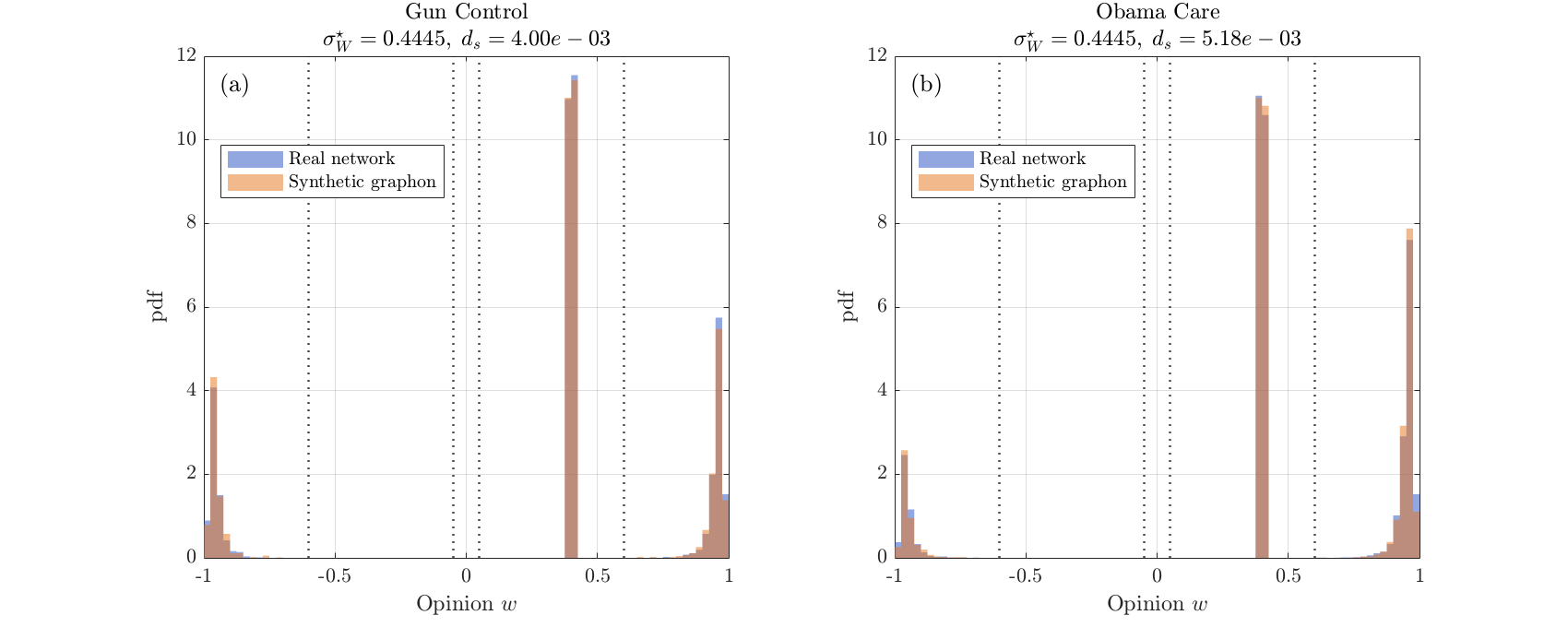}
\caption{Asymptotic opinion distributions of the calibrated synthetic graphon (orange) overlaid on those of the real network (blue), for (a) gun control and (b) Obamacare; dotted lines mark $w=0$ and the boundaries of $S$. The two distributions overlap almost exactly, with terminal Fourier distance $d_s=4.00\times10^{-3}$ for gun control and $d_s=5.18\times10^{-3}$ for Obamacare, both at the recovered bandwidth $\sigma_W^\star=0.4445$. Each distribution concentrates at the media target inside $S$, with residual clusters near the two poles corresponding to the deviance agents.}
\label{fig:surrogate-histograms}
\end{figure}

The experiment supports the following reading. Once the opinion-space mechanisms of compromise, diffusion, and media control on $S$ are fixed, the detailed topology of the interaction network influences the asymptotic opinion distribution essentially through its mean interaction strength, and a one-parameter Gaussian graphon, tuned to match that single statistic, recovers the macroscopic dynamics of the real network despite its structurally different, binary connectivity. This has a concrete methodological consequence: the surrogate requires only $\mathcal{O}(N)$ storage and no construction of the interaction network, while allowing the population size to be increased freely for smoother densities. The equivalence is at the level of the observable $f(w,t)$; on the real network, where the degree is heterogeneous, the position-resolved law is a mixture over agents of different connectivity, and the surrogate matches its opinion marginal rather than its microscopic structure.